\documentclass[a4paper,11pt]{article}
\pdfoutput=1 

\usepackage{jheppub} 

\usepackage[T1]{fontenc} 
\usepackage{subfig}
\usepackage{enumerate}

\makeatletter
\DeclareRobustCommand*{\bfseries}{%
   \not@math@alphabet\bfseries\mathbf
   \fontseries\bfdefault\selectfont
   \boldmath
}
\makeatother

\newcommand{\tr}{\text{tr}}

\newcommand{\Del}{\nabla}
\newcommand{\del}{\partial}
\newcommand{\dd}{\text{d}}

\renewcommand{\tilde}{\widetilde}
\renewcommand{\bar}{\overline}

\usepackage{bm}
\usepackage{tikz}
\usepackage{tikz-3dplot}
\usepackage{verbatim}
\usetikzlibrary{decorations.pathreplacing}
\usepackage{mathrsfs}

\usetikzlibrary{decorations.pathreplacing,decorations.markings}

\tikzset{
    >=stealth',
    punkt/.style={
           rectangle,
           rounded corners,
           draw=black, very thick,
           text width=6.5em,
           minimum height=2em,
           text centered},
    pil/.style={
           ->,
           thick,
           shorten <=2pt,
           shorten >=2pt,},
  on each segment/.style={
    decorate,
    decoration={
      show path construction,
      moveto code={},
      lineto code={
        \path [#1]
        (\tikzinputsegmentfirst) -- (\tikzinputsegmentlast);
      },
      curveto code={
        \path [#1] (\tikzinputsegmentfirst)
        .. controls
        (\tikzinputsegmentsupporta) and (\tikzinputsegmentsupportb)
        ..
        (\tikzinputsegmentlast);
      },
      closepath code={
        \path [#1]
        (\tikzinputsegmentfirst) -- (\tikzinputsegmentlast);
      },
    },
  },
  mid arrow/.style={postaction={decorate,decoration={
        markings,
        mark=at position .5 with {\arrow[#1]{stealth'}}
      }}}
}

 \newcommand{\ket}[1]{|#1\rangle}

\newcommand{\braket}[2]{\langle#1|#2\rangle}  

\newtheorem{theorem}{Theorem}

\newtheorem{definition}[theorem]{Definition}

\newtheorem{lemma}[theorem]{Lemma}

\newenvironment{proof}[1][Proof]{\noindent\textbf{#1.} }{\ \rule{0.5em}{0.5em}}
\newenvironment{argument}[1][Argument]{\noindent\textbf{#1.} }{\ \rule{0.5em}{0.5em}}

\usetikzlibrary{decorations.pathmorphing}
\tikzset{snake it/.style={decorate, decoration=snake}}

\title{Holographic scattering requires a connected entanglement wedge}

\author[a]{Alex May,}
\author[b]{Geoff Penington,}
\author[b]{and Jonathan Sorce}

\affiliation[a]{Department of Physics and Astronomy, University of British Columbia
6224 Agricultural Road, Vancouver, B.C., V6T 1W9, Canada}
\affiliation[b]{Stanford Institute for Theoretical Physics, Stanford University, 382 Via Pueblo Mall, Stanford, CA 94305-4060, U.S.A.}

\emailAdd{may@phas.ubc.ca}
\emailAdd{geoffp@stanford.edu}
\emailAdd{jsorce@stanford.edu}

\abstract{In AdS/CFT, there can exist local $2$-to-$2$ bulk scattering processes even when local scattering is not possible on the boundary; these have previously been studied in connection with boundary correlation functions. We show that boundary regions associated with these scattering configurations must have $O(1/G_N)$ mutual information, and hence a connected entanglement wedge. One of us previously argued for this statement from the boundary theory using operational tools in quantum information theory. We improve that argument to make it robust to small errors and provide a proof in the bulk using focusing arguments in general relativity. We also provide a direct link to entanglement wedge reconstruction by showing that the bulk scattering region must lie inside the connected entanglement wedge. Our construction implies the existence of nonlocal quantum computation protocols that are exponentially more efficient than the optimal protocols currently known.}

\begin{document} 
\maketitle
\flushbottom

\section{Introduction}

Asymptotically anti-de Sitter spacetimes in $(d+1)$ dimensions have the property that $2$-to-$d$ asymptotic scattering configurations may admit local scattering in the bulk but not on the boundary. In \cite{GGP, HPPS, Penedones, MSZ}, this feature was shown to be related to perturbative singularities in boundary correlation functions in holographic theories of quantum gravity. In this work, we relate this causal phenomenon to ideas from quantum information theory: whenever a $2$-to-$2$ asymptotic scattering configuration admits a bulk scattering region in an asymptotically AdS$_{2+1}$ spacetime, we show that boundary regions associated with the scattering process must have $O(1/G_N)$ mutual information in the holographic limit.

The existence or nonexistence of $O(1/G_N)$ mutual information between non-overlapping boundary regions is controlled by the bulk geometry. The HRRT formula \cite{RT, HRT, FLM, LM, DLR} states that the entanglement entropy of a boundary domain of dependence $V$ may be computed, at leading order in $G_N$, by the area of the smallest-area extremal surface homologous to the spatial boundary of $V$:
\begin{equation} \label{eq:HRRT}
    S(A) = \min_{\mathcal{E}_{V} \sim V} \left[ \frac{\text{Area}(\mathcal{E}_{V})}{4 G_N} \right] + O(1).
\end{equation}
The spacelike regions between cross-sections of $V$ and its HRRT surface $\mathcal{E}_{V}^{\text{min}}$ are called \emph{homology regions} $\mathcal{R}_{V}$. Their domain of dependence is called the entanglement wedge $E_W(V).$ It has been argued \cite{CKNR, HHLR, maximin, JLMS, DHW, noisyDHW} that boundary operators in $V$ are dual to bulk operators in $E_W(V)$; this is known as \emph{entanglement wedge reconstruction}. When the mutual information
\begin{equation}
    I(V_1 : V_2) = S(V_1) + S(V_2) - S(V_1 \cup V_2)
\end{equation}
between two spacelike-separated domains of dependence $V_1$ and $V_2$ is nonzero at order $1/G_N$, equation \eqref{eq:HRRT} implies that $E_W(V_1 \cup V_2)$ is connected.\footnote{Here we implicitly assume that $V_1$ and $V_2$ are themselves connected; when they have multiple connected components, the entanglement wedge need only contain a connected component that ``stretches between'' components of $V_1$ and $V_2$.}

Suppose we choose asymptotic boundary `input points' $\{c_1, c_2\}$ and `output points' $\{r_1, r_2\}$, such that local $2$-to-$2$ scattering is possible in the bulk, but not on the boundary. We will show that there must exist large, $O(1/G_N)$ mutual information between the boundary regions $V_1$ and $V_2$, which are defined respectively as the intersection of the pasts of \emph{both} output points and the future of the corresponding input point $c_1$ or $c_2$. This is sketched in Figure \ref{fig:phasetransition3d}.

\begin{figure}
    \centering
    \subfloat[\label{fig:connected-scattering}]{
    \tdplotsetmaincoords{15}{0}
    \begin{tikzpicture}[scale=1.4,tdplot_main_coords]
    \tdplotsetrotatedcoords{0}{30}{0}
    \draw[gray] (-2,-0.25,0) -- (-2,6.25,0);
    \draw[gray] (2,-0.25,0) -- (2,6.25,0);
    
    \begin{scope}[tdplot_rotated_coords]
    
    \draw[domain=0:45,variable=\x,smooth, fill=black!60!,opacity=0.8] plot ({-2*sin(\x)}, {1+\x/45}, {2*cos(\x)}) -- plot ({-2*sin((45-\x))}, {3-(45-\x)/45}, {2*cos(45-\x)}) --  plot ({2*sin(\x)}, {3-\x/45}, {2*cos(\x)}) -- plot ({2*sin(45-\x)}, {1+(45-\x)/45}, {2*cos(45-\x)});
    
    \draw[domain=0:45,variable=\x,smooth,thick] plot ({-2*sin(\x)}, {1+\x/45}, {2*cos(\x)});
    \draw[domain=0:45,variable=\x,smooth,thick] plot ({2*sin(\x)}, {1+\x/45}, {2*cos(\x)});
    \draw[domain=0:45,variable=\x,smooth,thick] plot ({-2*sin(\x)}, {3-\x/45}, {2*cos(\x)});
    \draw[domain=0:45,variable=\x,smooth,thick] plot ({2*sin(\x)}, {3-\x/45}, {2*cos(\x)});
    
    \begin{scope}[canvas is xz plane at y=-0.25]
    \draw[gray] (0,0) circle[radius=2] ;
    \end{scope}
    
    \begin{scope}[canvas is xz plane at y=6.25]
    \draw[gray] (0,0) circle[radius=2] ;
    \end{scope}
    
    \draw[red] (0,1,-2) -- (0,2,-1);
    \draw[red] (0,1,2) -- (0,2,1);
    
    \begin{scope}[canvas is xz plane at y=2]
    
    \draw[gray] (0,0) circle (2);
    
    \draw [domain=-45:45,fill=lightgray,opacity=0.8] plot ({2*cos(\x+90)}, {2*sin(\x+90)}) -- (-1.41,1.41) to [out=-45,in=45] (-1.41,-1.41) -- plot ({2*cos(\x-90)}, {2*sin(\x-90)}) -- (1.41,-1.41) to [out=135,in=-135] (1.41,1.41);
    
    \draw [green,ultra thick,domain=-45:45] plot ({2*cos(\x+90)}, {2*sin(\x+90)});
    
    \draw[blue,thick] (1.41,1.41) to [out=-135,in=+135] (1.41,-1.41);
    \draw[blue,thick] (-1.41,1.41) to [out=-45,in=45] (-1.41,-1.41);
    
    \end{scope}
    
    \draw[domain=0:90,variable=\x,smooth,dashed] plot ({2*sin(\x+180)}, {3+\x/45}, {2*cos(\x+180)});
    \draw[domain=0:90,variable=\x,smooth,dashed] plot ({2*sin(\x+180)}, {3+\x/45}, {-2*cos(\x+180)});
    \draw[domain=0:90,variable=\x,smooth,dashed] plot ({-2*sin(\x+180)}, {3+\x/45}, {2*cos(\x+180)});
    \draw[domain=0:90,variable=\x,smooth,dashed] plot ({-2*sin(\x+180)}, {3+\x/45}, {-2*cos(\x+180)});
    
    \draw[thick, red,-triangle 45] (0,3,0) -- (2,5,0);
    \draw[thick, red,-triangle 45] (0,3,0) -- (-2,5,0);
    
    \draw[thick,red,-triangle 45] (0,2,-1) -- (0,3,0);
    \draw[thick,red,-triangle 45] (0,2,1) -- (0,3,0);
    
    \draw plot [mark=*, mark size=1.5] coordinates{(2,5,0)};
    \node[above] at (2,5,0) {$r_1$};
    \draw plot [mark=*, mark size=1.5] coordinates{(-2,5,0)};
    \node[above] at (-2,5,0) {$r_2$};
    
    \draw[domain=0:45,variable=\x,smooth, fill=black!50!,opacity=0.8] plot ({-2*sin(\x+180)}, {1+\x/45}, {2*cos(\x+180)}) -- plot ({-2*sin((45-\x)+180)}, {3-(45-\x)/45}, {2*cos(45-\x+180)}) --  plot ({2*sin(\x+180)}, {3-\x/45}, {2*cos(\x+180)}) -- plot ({2*sin(45-\x+180)}, {1+(45-\x)/45}, {2*cos(45-\x+180)});
    
    \begin{scope}[canvas is xz plane at y=2]
    \draw [green,ultra thick,domain=-45:45] plot ({2*cos(\x-90)}, {2*sin(\x-90)});
    \end{scope}
    
    \draw[domain=0:45,variable=\x,smooth,thick] plot ({-2*sin(\x+180)}, {1+\x/45}, {2*cos(\x+180)});
    \draw[domain=0:45,variable=\x,smooth,thick] plot ({2*sin(\x+180)}, {1+\x/45}, {2*cos(\x+180)});
    \draw[domain=0:45,variable=\x,smooth,thick] plot ({-2*sin(\x+180)}, {3-\x/45}, {2*cos(\x+180)});
    \draw[domain=0:45,variable=\x,smooth,thick] plot ({2*sin(\x+180)}, {3-\x/45}, {2*cos(\x+180)});
    
    \draw plot [mark=*, mark size=1.5] coordinates{(0,1,-2)};
    \node[left] at (0,1,-2) {$c_1$};
    \draw plot [mark=*, mark size=1.5] coordinates{(0,1,2)};
    \node[below] at (0,1,2) {$c_2$};
    \draw plot [mark=*, mark size=1.5] coordinates{(0,3,0)};
    \node[left] at (0,3,0) {$p$};
    
    \end{scope}
    \end{tikzpicture}
    }
    \hfill
    \subfloat[\label{fig:disconnected-scattering}]{
    \tdplotsetmaincoords{13}{0}
    \begin{tikzpicture}[scale=1.4,tdplot_main_coords]
    \tdplotsetrotatedcoords{0}{25}{0}
    \draw[gray] (-2,-0.25,0) -- (-2,6.25,0);
    \draw[gray] (2,-0.25,0) -- (2,6.25,0);
    \begin{scope}[tdplot_rotated_coords]
    
    \draw[domain=0:30,variable=\x,smooth,fill=black!60!,opacity=0.8] plot ({2*sin(\x)}, {1.34+\x/45}, {2*cos(\x)}) -- plot ({2*sin(30-\x)}, {2.66-(30-\x)/45}, {2*cos(30-\x)}) -- plot ({-2*sin(\x)}, {2.66-\x/45}, {2*cos(\x)}) -- plot ({-2*sin(30-\x)}, {1.34+(30-\x)/45}, {2*cos(30-\x)});
    
    \begin{scope}[canvas is xz plane at y=-0.25]
    \draw[gray] (0,0) circle[radius=2] ;
    \end{scope}
    
    \begin{scope}[canvas is xz plane at y=6.25]
    \draw[gray] (0,0) circle[radius=2] ;
    \end{scope}
    
    \begin{scope}[canvas is xz plane at y=2]
    
    \draw[gray] (0,0) circle (2);
    
    \draw [domain=60:120,fill=lightgray,opacity=0.8] plot ({2*cos(\x)}, {2*sin(\x)}) -- (-1, 1.73) to [out=-60,in=-120] (1, 1.73);
     \draw [domain=-60:-120,fill=lightgray,opacity=0.8] plot ({2*cos(\x)}, {2*sin(\x)}) -- (-1, -1.73) to [out=60,in=120] (1, -1.73);
    
    \draw [green,ultra thick,domain=60:120] plot ({2*cos(\x)}, {2*sin(\x)});
    
    \draw[blue, thick] (1, 1.73) to [out=-120,in=-60] (-1, 1.73);
    \draw[blue, thick] (1, -1.73) to [out=120,in=60] (-1, -1.73);
    
    \end{scope}
    
    \draw[domain=0:30,variable=\x,smooth,thick] plot ({2*sin(\x)}, {1.34+\x/45}, {2*cos(\x)});
    \draw[domain=0:30,variable=\x,smooth,thick] plot ({-2*sin(\x)}, {1.34+\x/45}, {2*cos(\x)});
    \draw[domain=0:30,variable=\x,smooth,thick] plot ({2*sin(\x)}, {2.66-\x/45}, {2*cos(\x)});
    \draw[domain=0:30,variable=\x,smooth,thick] plot ({-2*sin(\x)}, {2.66-\x/45}, {2*cos(\x)});
    
    \draw[red,-triangle 45] (0,1.34,2) -- (0,3.34,0);
    \draw[red,-triangle 45] (0,1.34,-2) -- (0,3.34,0);
    
    \draw[red,-triangle 45] (0,3.34,0) -- (2,5.34,0);
    \draw[red,-triangle 45] (0,3.34,0) -- (-2,5.34,0);
    
    \draw[domain=0:90,variable=\x,smooth,dashed] plot ({2*sin(\x)}, {2.66+\x/45}, {-2*cos(\x)});
    \draw[domain=0:90,variable=\x,smooth,dashed] plot ({2*sin(\x)}, {2.66+\x/45}, {2*cos(\x)});
    
    \draw[domain=0:90,variable=\x,smooth,dashed] plot ({-2*sin(\x)}, {2.66+\x/45}, {-2*cos(\x)});
    \draw[domain=0:90,variable=\x,smooth,dashed] plot ({-2*sin(\x)}, {2.66+\x/45}, {2*cos(\x)});
    
    \draw[domain=0:30,variable=\x,smooth,fill=black!50!,opacity=0.8] plot ({2*sin(\x+180)}, {1.34+\x/45}, {2*cos(\x+180)}) -- plot ({2*sin(180+30-\x)}, {2.66-(30-\x)/45}, {2*cos(180+30-\x)}) -- plot ({-2*sin(180+\x)}, {2.66-\x/45}, {2*cos(180+\x)}) -- plot ({-2*sin(180+30-\x)}, {1.34+(30-\x)/45}, {2*cos(180+30-\x)});
    
    \draw plot [mark=*, mark size=1.5] coordinates{(0,1.34,2)};
    \node[below] at (0,1.34,2) {$c_2$};
    \draw plot [mark=*, mark size=1.5] coordinates{(0,1.34,-2)};
    \node[below] at (0,1.34,-2) {$c_1$};
    \draw plot [mark=*, mark size=1.5] coordinates{(0,3.34,0)};
    \node[right] at (0.1,3.45,0) {$p$};
    
    \draw plot [mark=*, mark size=1.5] coordinates{(2,4.66,0)};
    \node[above] at (2,4.66,0) {$r_1$}; 
    \draw plot [mark=*, mark size=1.5] coordinates{(-2,4.66,0)};
    \node[above] at (-2,4.66,0) {$r_2$}; 
    
    \begin{scope}[canvas is xz plane at y=2]
    \draw [green,ultra thick,domain=-60:-120] plot ({2*cos(\x)}, {2*sin(\x)});
    \end{scope}
    
    \draw[domain=0:30,variable=\x,smooth,thick] plot ({2*sin(\x+180)}, {1.34+\x/45}, {2*cos(\x+180)});
    \draw[domain=0:30,variable=\x,smooth,thick] plot ({-2*sin(\x+180)}, {1.34+\x/45}, {2*cos(\x+180)});
    \draw[domain=0:30,variable=\x,smooth,thick] plot ({2*sin(\x+180)}, {2.66-\x/45}, {2*cos(\x+180)});
    \draw[domain=0:30,variable=\x,smooth,thick] plot ({-2*sin(\x+180)}, {2.66-\x/45}, {2*cos(\x+180)});
    
    \end{scope}
    \end{tikzpicture}
    }
    
    \caption{(a) When a set of four boundary points $\{c_1, c_2, r_1, r_2\}$ has a bulk scattering region, the entanglement wedge of associated domains of dependence is connected. (b) When there is no scattering region in the bulk, the entanglement wedge need not be connected. See Theorem \ref{thm:main} for a formal statement. This figure is reproduced from \cite{may2019quantum}.}
    \label{fig:phasetransition3d}
\end{figure}
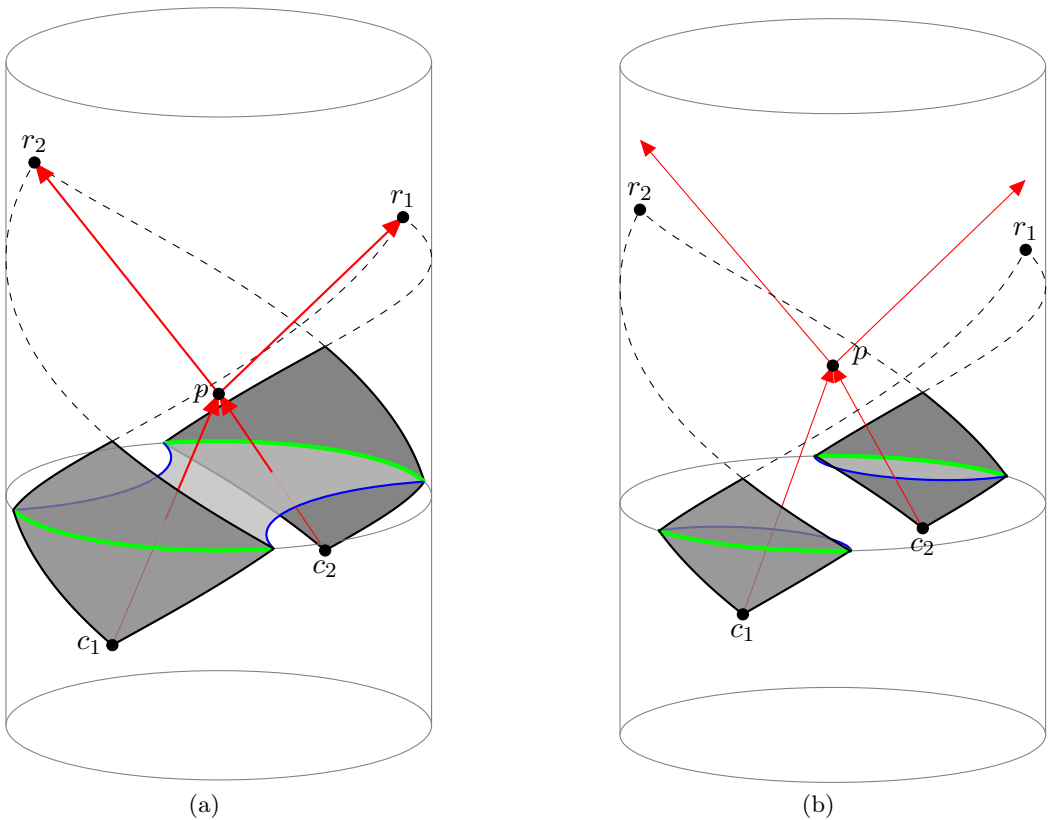

More formally, we define the {bulk scattering region} 
\begin{equation}
    J_{1 2 \rightarrow 12}
        \equiv J^+(c_1) \cap J^+(c_2) \cap J^-(r_1) J^-(r_2).
\end{equation}
as the set of all bulk points that are in the future of both input points and the past of both output points. The corresponding region on the boundary is denoted
\begin{equation}
    \hat{J}_{12 \rightarrow 12}
        \equiv \hat{J}^+(c_1) \cap \hat{J}^+(c_2) \cap \hat{J}^-(r_1) \cap \hat{J}^-(r_2).
\end{equation}
In general, we use $J^{\pm}(S)$ to denote the future/past of a set $S$ in the bulk, and $\hat{J}^{\pm}(S)$ to denote the future/past of $S$ restricted to the boundary. We also use the restricted notation, e.g., $J_{1 \rightarrow 12}$ and $\hat{J}_{1 \rightarrow 12}$ to denote the bulk/boundary points in the future of $c_1$ and the past of both $r_1$ and $r_2$, as above. A \emph{bulk-only scattering configuration} is one in which $J_{1 2 \rightarrow 12}$ has nonempty interior, while $\hat{J}_{1 2 \rightarrow 1 2}$ is empty. Our main theorem is then as follows:
\begin{theorem}[Connected wedge theorem] \label{thm:main}
    Let $\{c_1, c_2, r_1, r_2\}$ be a bulk-only scattering configuration on the boundary of an asymptotically AdS spacetime with a holographic dual. Let $V_1$ and $V_2$ be boundary regions defined by
    \begin{equation}
        V_1 = \hat{J}_{1 \rightarrow 12} \quad \text{and} \quad V_2 = \hat{J}_{2 \rightarrow 12}.
    \end{equation}
    Then $V_1$ and $V_2$ are spacelike-separated domains of dependence, and the entanglement wedge of $V_1 \cup V_2$ is connected.
\end{theorem}

When $\{c_1, c_2, r_1, r_2\}$ admits a boundary scattering region, $V_1$ and $V_2$ overlap on the boundary and their union $V_1 \cup V_2$ automatically has a connected entanglement wedge. The converse of Theorem \ref{thm:main}, that a connected entanglement wedge implies a nonempty scattering region, is not true.\footnote{For discussion of this point and an explicit counterexample, see Section \ref{sec:null-membrane}.}

Theorem \ref{thm:main} was recently suggested by one of us based on a boundary argument using tools from operational quantum information theory \cite{may2019quantum}. The purpose of the present work is (i) to refine that argument by accounting for errors in the quantum information protocol, (ii) to provide a proof of Theorem \ref{thm:main} in the bulk using semiclassical general relativity together with the HRRT formula, and (iii) to interpret Theorem \ref{thm:main} in terms of recent developments in quantum information theory applied to quantum gravity. 

The arguments of \cite{may2019quantum} suggested that general relativity needs to have a sophisticated understanding of quantum information theory for the AdS/CFT correspondence to be self-consistent. In the present work, we show that this is indeed the case: general relativity does know quantum information. 

This is similar in spirit to earlier work on deriving Einstein's equations from the first law of entanglement \cite{EE-EE1, EE-EE2, EE-EE3, EE-EE4, EE-EE5, EE-EE6}. However, in our case, information-theoretic boundary arguments imply bulk constraints that are highly \emph{nonlocal} in both space and time. The existence of a bulk scattering region in the far future leads to an inequality relating the areas of two surfaces, themselves separated by a large spacelike distance. Nevertheless, general relativity is clever enough to enforce this relationship.

We now give a brief summary of this paper.

In \hyperref[sec:taskssection]{{\bf section two}},
we give a refined version of the arguments toward Theorem \ref{thm:main} from \cite{may2019quantum}. By thinking of the scattering process as a quantum computation that needs to be possible in both the bulk and the boundary, we are able to use operational tools from quantum information theory to argue that $V_1$ and $V_2$ must have $O(1/G_N)$ mutual information. Specifically, we argue that the scattering process can be used to perform a simple quantum computation; conjectured restrictions on the resources necessary to complete the same computation on the boundary give the desired result.

In \hyperref[sec:GR]{{\bf section three}}, we prove Theorem \ref{thm:main} from a bulk perspective using classical general relativity. This proof proceeds by contradiction. We assume that the entanglement wedge of $V_1 \cup V_2$ is disconnected, and then, by applying focusing theorems both to lightsheets of extremal surfaces and to causal horizons, construct a surface with area smaller than $\mathcal{E}_{V_1 \cup V_2}^{\text{min}}$ on any complete achronal slice containing $\mathcal{E}_{V_1 \cup V_2}^{\text{min}}$. This contradicts the maximin procedure for finding HRRT surfaces \cite{maximin}. This proof has the additional feature of providing a lower bound on $I(V_1 : V_2)$ in terms of the geometry of bulk null surfaces.

In \hyperref[sec:wedge]{{\bf section four}}, we prove that whenever the conditions of Theorem \ref{thm:main} are satisfied, the scattering region $J_{12\rightarrow 12}$ must lie inside the entanglement wedge $E_W(V_1 \cup V_2)$.\footnote{We thank Veronika Hubeny and Mukund Rangamani for first suggesting to us that this might be true.} Intuitively, this means that the bulk scattering point is always encoded in the boundary region $\hat{J}_{1 \rightarrow 12} \cup \hat{J}_{2 \rightarrow 12}$, which is the region that can be causally influenced by at least one of the two inputs points $c_i$, and can itself causally influence both output points $r_i$.

In \hyperref[sec:discussion]{{\bf section five}}, we discuss various related issues and open questions. In particular, we point out that there exist quantum computations where the most efficient known protocols require an exponential amount of entanglement, when no scattering region exists. Nonetheless, holography appears to be able to perform the computation with only linear entanglement in the boundary theory. This is consistent with all known lower bounds on the entanglement required, but provides strong evidence that vastly more efficient protocols exist than are currently known, with consequences for quantum cryptography. We also discuss potential generalizations of these ideas to higher-dimensional settings, alternative proofs of Theorem \ref{thm:main}, and a potential connection to metric reconstruction.

Throughout this paper, we work in units with $\hbar = c = 1$, leaving Newton's constant $G_N$ explicit. We also often work in units where the AdS radius satisfies $\ell_{AdS}=1.$ Section \ref{sec:GR} uses the notational conventions of \cite{waldbook} for general relativity, working in a mostly-pluses metric signature and using Latin letters $a, b, \dots$ to denote ``abstract'' tensor indices.

\section{Boundary proof from quantum tasks}\label{sec:taskssection}

\subsection{The \texorpdfstring{$\textbf{B}_{84}$}{TEXT} task} 

In the introduction we discussed scattering experiments, and insights gained by noting that some arrangements of input and output points give bulk-only scattering configurations. In this section we view such scattering experiments in the context of quantum information --- the in- and out- scattering states become the inputs and outputs of a quantum computation. As we will see, bulk-only scattering configurations place requirements on the boundary state. In particular, arguments from operational quantum information theory suggest that large correlations must be present between the regions $V_1$ and $V_2$ that are associated with the scattering configuration. These large correlations imply a connected entanglement wedge.

To motivate the particular scattering set-up we will use to argue for Theorem \ref{thm:main}, it is helpful to revisit quantum teleportation. Consider Figure \ref{fig:sendreceivetask}, which shows what we call the send-receive task. We consider an input point $c$, where a quantum state $\ket{\psi}$ is received, and an output point $r$, where $\ket{\psi}$ should be sent. Importantly, there are two approaches to completing this send-receive task --- a local approach and a nonlocal approach. The local approach is to record $\ket{\psi}$ in some localized excitation and send it from $c$ to $r$. The nonlocal approach is to use entanglement to teleport $\ket{\psi}$ from $c$ to $r$. In the context of holographic traversable wormholes \cite{gao2017traversable,susskind2018teleportation,maldacena2017diving}, from a bulk perspective a localized excitation travels through a wormhole --- i.e. the local approach --- while in the boundary picture the excitation is transmitted via a version of quantum teleportation --- the nonlocal approach. 

\begin{figure}
    \centering
    \subfloat[\label{fig:sendreceivetask}]{
    \begin{tikzpicture}[scale=0.5]
    
    \node[below left] at (-4,0) {$c$};
    \draw[fill=black] (-4,0) circle (0.15);
    \draw[->] (-4,-1) -- (-4,-0.1);
    \node[below] at (-4,-1) {$\ket{\psi}$};

    \draw[->] (4,8) -- (4,9);
    \node[below right] at (4,8) {$r$};
    \draw[fill=blue] (4,8) circle (0.15);
    \node[above] at (4,9) {$\ket{\psi}$};
    
    \draw[->] (-7,-2) -- (-7,-0.5);
    \draw[->] (-7,-2) -- (-5.5,-2);
    \node[below] at (-5.5,-2) {$x$};
    \node[right] at (-7,-0.5) {$t$};
    
    \end{tikzpicture}
    }
    \hfill
    \subfloat[\label{fig:B84taskschematic}]{
    \begin{tikzpicture}[scale=0.5]
    
    \node[below left] at (-4,0) {$c_1$};
    \draw[fill=black] (-4,0) circle (0.15);
    \draw[->] (-4,-1) -- (-4,-0.1);
    \node[below] at (-4,-1) {$H^q\ket{b}$};

    \node[below right] at (4,0) {$c_2$};
    \draw[fill=black] (4,0) circle (0.15);
    \draw[->] (4,-1) -- (4,-0.1);
    \node[below] at (4,-1) {$q$};

    \draw[->] (4,8) -- (4,9);
    \node[below right] at (4,8) {$r_2$};
    \draw[fill=blue] (4,8) circle (0.15);
    \node[above] at (4,9) {$b$};

    \draw[->] (-4,8) -- (-4,9);
    \node[below left] at (-4,8) {$r_1$};
    \draw[fill=blue] (-4,8) circle (0.15);
    \node[above] at (-4,9) {$b$};
    
    \draw[->] (-7,-2) -- (-7,-0.5);
    \draw[->] (-7,-2) -- (-5.5,-2);
    \node[below] at (-5.5,-2) {$x$};
    \node[right] at (-7,-0.5) {$t$};

    \end{tikzpicture}
    }
    \caption{a) The send-receive task. A quantum state $\ket{\psi}$ is received at a spacetime location $c$, and must be returned to $r$. b) The $\textbf{B}_{84}$ task, which we employ to prove the connected wedge theorem. At $c_1$ the quantum system $A$ is received which holds a state $H^q\ket{b}$. At $c_2$ a classical bit $q\in\{0,1\}$ is received. For the task to be completed successfully, $b$ should be produced at both $r_1$ and $r_2$.}
    \label{fig:twotasks}
\end{figure}
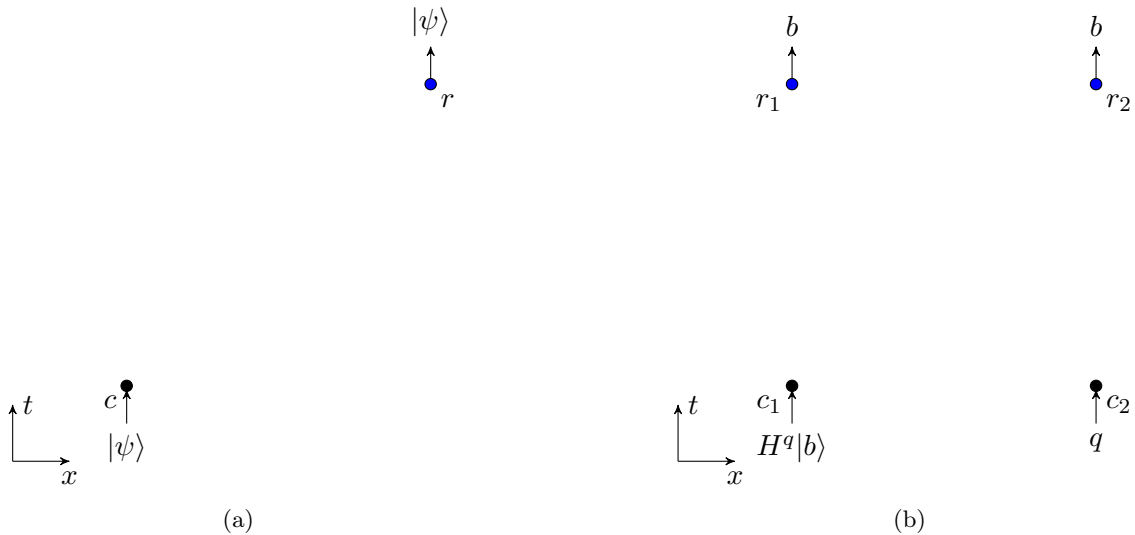

Now consider the set-up shown in Figure \ref{fig:B84taskschematic}. For historical reasons \cite{bennett1984proceedings}, we refer to this as the $\textbf{B}_{84}$ task. The task has two inputs and two outputs. The input locations are labeled by $c_1$, $c_2$, the output locations by $r_1$, $r_2$. At $c_2$, a classical bit $q \in \{0,1\}$ is received. At $c_1$, a qubit in the state $H^q\ket{b}$ is received, where $b\in \{0,1\}$ and $H$ is the Hadamard operator\footnote{The Hadamard operator satisfies $H\ket{0} = \ket{+}$ and $H\ket{1} = \ket{-}$, with $\ket{\pm}= (\ket{0}\pm \ket{1})/\sqrt{2}$.}. We consider the input parameters $b$ and $q$ to be drawn independently and at random. To complete the task, $b$ should be made available at $r_1$ and at $r_2$.

As with the send-receive task, there are local and nonlocal procedures for completing the $\textbf{B}_{84}$ task. We will argue that, as for traversable wormholes, the task can be completed locally in the bulk, but must be completed nonlocally on the boundary. The local approach to completing $\textbf{B}_{84}$ is straightforward. Send $q$ from $c_2$ to $J_{12\rightarrow 12}$, and $H^q\ket{b}$ from $c_1$ to $J_{12\rightarrow 12}$. Inside $J_{12\rightarrow 12}$, apply $H^q$ and then measure in the computational basis to determine $b$. Finally, send $b$ from $J_{12\rightarrow 12}$ to both $r_1$ and $r_2$. We can use this simple protocol to complete the $\textbf{B}_{84}$ task whenever the spacetime region $J_{12\rightarrow 12}$ is nonempty. 

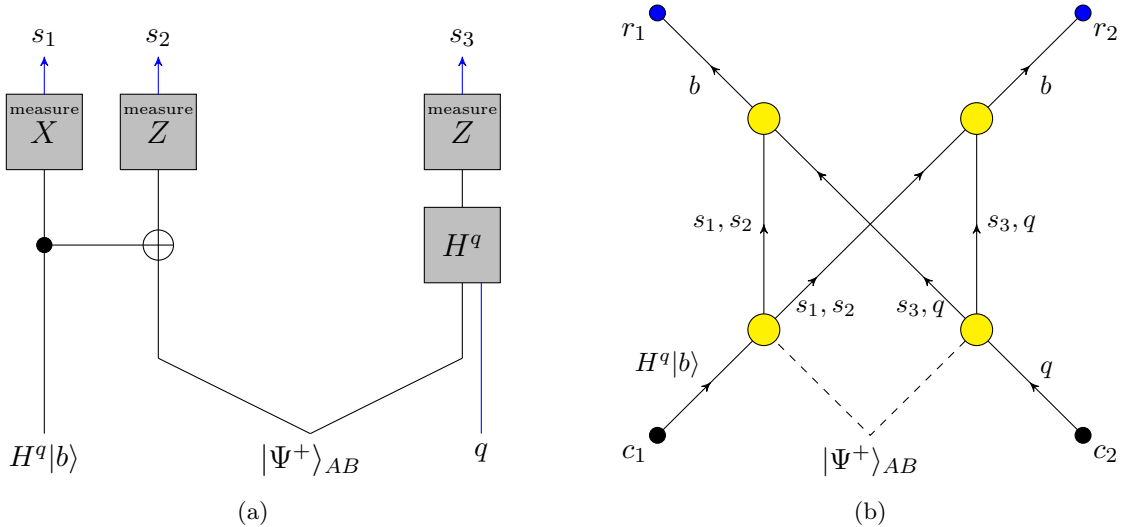
\begin{figure}
    \centering
    \subfloat[\label{fig:nonlocalcircuit}]{
    \begin{tikzpicture}[node distance = 2cm, auto,rotate=90]
    
    \draw (0,0) -- (1,2);
    \draw (0,0) -- (1,-2);
    \node[below] at (0,0) {\large{$\ket{\Psi^+}_{AB}$}};
    
    \draw (1,-2) -- (2,-2);
    
    \draw[fill=lightgray] (2,-2.5) -- (2,-1.5) -- (3,-1.5) -- (3,-2.5) -- (2,-2.5); 
    \node at (2.5,-2) {\large{$H^q$}};
    
    \draw (3,-2) -- (3.5,-2);
    
    \draw[fill=lightgray] (3.5,-2.5) -- (3.5,-1.5) -- (4.5,-1.5) -- (4.5,-2.5) -- (3.5,-2.5);
    \node[above] at (4.1,-2) {\tiny{measure}};
    \node at (4,-2) {\large{$Z$}};
    
    \draw (1,2) -- (3.5,2);
    \draw (0,3.5) -- (3.5,3.5);
    
    \node[below] at (0,3.5) {$H^q\ket{b}$};
    
    \draw[fill=black] (2.5,3.5) circle (0.1);
    \draw (2.5,3.5) -- (2.5,1.8);
    \draw (2.5,2) circle (0.2);
    
    \draw[fill=lightgray] (3.5,4) -- (4.5,4) -- (4.5,3) -- (3.5,3) -- (3.5,4);
    \node[above] at (4.1,3.5) {\tiny{measure}};
    \node at (4,3.5) {\large{$X$}};
    
    \draw[fill=lightgray] (3.5,2.5) -- (4.5,2.5) -- (4.5,1.5) -- (3.5,1.5) -- (3.5,2.5);
    \node[above] at (4.1,2) {\tiny{measure}};
    \node at (4,2) {\large{$Z$}};
    
    \draw[blue,->] (4.5,2) -- (5,2);
    \node[above] at (5,2) {$s_2$};
    
    \draw[blue,->] (4.5,3.5) -- (5,3.5);
    \node[above] at (5,3.5) {$s_1$};
    
    \draw[blue,->] (4.5,-2) -- (5,-2);
    \node[above] at (5,-2) {$s_3$};
    
    \draw[blue] (0,-2.25) -- (2,-2.25);
    \node[below] at (0,-2.25) {$q$};
    
    \end{tikzpicture}
    }
    \hfill
    \subfloat[\label{fig:nonlocalschematic}]{
    \begin{tikzpicture}[scale=0.7]
    
    \draw[postaction={on each segment={mid arrow}}] (-4,0) -- (-2,2) -- (-2,6) -- (-4,8);
    \draw[postaction={on each segment={mid arrow}}] (4,0) -- (2,2) -- (2,6) -- (4,8);
    \draw[postaction={on each segment={mid arrow}}] (-2,2) -- (0,4) -- (2,6);
    \draw[postaction={on each segment={mid arrow}}] (2,2) -- (0,4) -- (-2,6);
    
    \draw[dashed] (2,2) -- (0,0) -- (-2,2);
    \node[below] at (0,0) {$\ket{\Psi^+}_{AB}$};
    
    \draw[fill=yellow] (-2,2) circle (0.3);
    \draw[fill=yellow] (2,2) circle (0.3);
    \draw[fill=yellow] (-2,6) circle (0.3);
    \draw[fill=yellow] (2,6) circle (0.3);
    
    \node[below left] at (-4,0) {$c_1$};
    \draw[fill=black] (-4,0) circle (0.15);

    \node[below right] at (4,0) {$c_2$};
    \draw[fill=black] (4,0) circle (0.15);

    \node[below right] at (4,8) {$r_2$};
    \draw[fill=blue] (4,8) circle (0.15);

    \node[below left] at (-4,8) {$r_1$};
    \draw[fill=blue] (-4,8) circle (0.15);
    
    \node[above left] at (-3,0.9) {\small{$H^q\ket{b}$}};
    \node[above right] at (3,0.9) {\small{$q$}};
    
    \node[below left] at (-3,7) {\small{$b$}};
    \node[below right] at (3,7) {\small{$b$}};
    
    \node[left] at (-2,4) {\small{$s_1,s_2$}};
    \node[right] at (2,4) {\small{$s_3,q$}};
    
    \node[right] at (-1.6,2.4) {\small{$s_1,s_2$}};
    \node[left] at (1.6,2.4) {\small{$s_3,q$}};
    
    \end{tikzpicture}
    }
    \caption{Circuit diagram that completes the $\mathbf{B}_{84}$ task. Blue lines indicate classical inputs and outputs. The protocol uses one EPR pair, $\ket{\Psi^+}_{AB}=\frac{1}{\sqrt{2}}(\ket{00}+\ket{11})$ as a resource. $b$ is a function of the classical measurement outcomes according to $(-1)^b= s_1^qs_2^{1-q}s_3$. Importantly, the operations performed in the circuit can be placed in the nonlocal form shown in b).}
\end{figure}

To complete the $\textbf{B}_{84}$ task using a nonlocal approach, consider the quantum circuit shown in Figure \ref{fig:nonlocalcircuit}. It is straightforward to check that this circuit completes the task. There are two important features of this circuit. First, it can be run without ever performing a gate within the region $J_{12\rightarrow 12}$. This is why we say that it completes the task nonlocally. Second, the circuit makes use of an EPR pair shared between $c_1$ and $c_2$. This shared entanglement is necessary to complete the task using a nonlocal approach. 

Before studying the $\textbf{B}_{84}$ task in more detail we should make explicit its connection to holography. Consider Figure \ref{fig:bulkandboundaryB84}. There, we have embedded the $\textbf{B}_{84}$ task into a scattering scenario in an asymptotically AdS$_{2+1}$ spacetime. Specifically, the points $c_1,c_2,r_1,r_2$ of the $\textbf{B}_{84}$ task become points on the boundary of AdS. The inputs $q,H^q\ket{b}$ are recorded into localized excitations falling in towards the center of AdS, and the outputs are localized excitations coming out towards the boundary. We can consider this task within either the bulk or boundary descriptions. Since the AdS/CFT duality relates the two viewpoints, any process that successfully completes the task in the bulk must be dual to a corresponding boundary process that also successfully completes the task. 

\begin{figure}
    \centering
    \subfloat[\label{subfig:B84bulk}]{
    \tdplotsetmaincoords{10}{0}
    \begin{tikzpicture}[scale=1.0,tdplot_main_coords]
    \tdplotsetrotatedcoords{0}{20}{0}
    \draw (-2,0,0) -- (-2,4,0);
    \draw (2,0,0) -- (2,4,0);
    
    \begin{scope}[tdplot_rotated_coords]
    \begin{scope}[canvas is xz plane at y=0]
    \draw (0,0) circle [radius=2];
    \end{scope}
    
    \begin{scope}[canvas is xz plane at y=4]
    \draw (0,0) circle [radius=2];
    \end{scope}
    
    \draw plot [mark=*, mark size=2] coordinates{(2,0,0)};
    \node[right] at (2,0,0) {$c_2$};
    \draw[->] (2,-0.5,0) -- (2,0,0);
    \node[below] at (2,-0.5,0) {$q$};
    
    \draw plot [mark=*, mark size=2] coordinates{(-2,0,0)};
    \node[left] at (-2,0,0) {$c_1$};
    \draw[->] (-2,-0.5,0) -- (-2,0,0);
    \node[below] at (-2,-0.5,0) {$H^q\ket{b}$};
    
    \draw[blue] plot [mark=*, mark size=2] coordinates{(0,4,-2)};
    \node[below left] at (0,4,-2) {$r_1$};
    \draw[->] (0,4,-2) -- (0,4.5,-2);
    \node[above] at (0,4.5,-2) {$b$};
    
    \draw[blue] plot [mark=*, mark size=2] coordinates{(0,4,2)};
    \node[below right] at (0,4,2) {$r_2$};
    \draw[->] (0,4,2) -- (0,4.5,2);
    \node[above] at (0,4.5,2) {$b$};
    
    \end{scope}
    \end{tikzpicture} 
    }
    \hfill
    \subfloat[\label{subfig:B84boundary}]{
    \begin{tikzpicture}[scale=1.25]
    
    \draw (-2,0) -- (2,0) -- (2,2) -- (-2,2) -- (-2,0);
    
    \draw[black] plot [mark=*, mark size=2] coordinates{(-2,0)};
    \draw[->] (-2,-0.5)--(-2,-0.05);
    \node[below] at (-2,-0.5) {$q$};
    \node[below left] at (-2,0) {$c_2$};
    
    \draw[black] plot [mark=*, mark size=2] coordinates{(2,0)};
    \node[below right] at (2,0) {$c_2$};
    
    \draw[black] plot [mark=*, mark size=2] coordinates{(0,0)};
    \draw[->] (0,-0.5)--(0,-0.05);
    \node[below] at (0,-0.5) {$H^q\ket{b}$};
    \node[below left] at (0,0) {$c_1$};
    
    \draw[->] (-1,2) -- (-1,2.5);
    \draw[blue] plot [mark=*, mark size=2] coordinates{(-1,2)};
    \node[above right] at (-1,2) {$r_1$};
    
    \draw[->] (1,2) -- (1,2.5);
    \draw[blue] plot [mark=*, mark size=2] coordinates{(1,2)};
    \node[above right] at (1,2) {$r_2$};
    
    \node at (0,-1.5) {$ $};
    \node at (0,3.5) {$ $};
    
    \end{tikzpicture}
    }
    \caption{The $\mathbf{B}_{84}$ task considered in an asymptotically AdS spacetime. The bulk view is shown at left, while the boundary perspective is at right. The task should be possible in the boundary whenever it is possible in the bulk.}
    \label{fig:bulkandboundaryB84}
\end{figure}
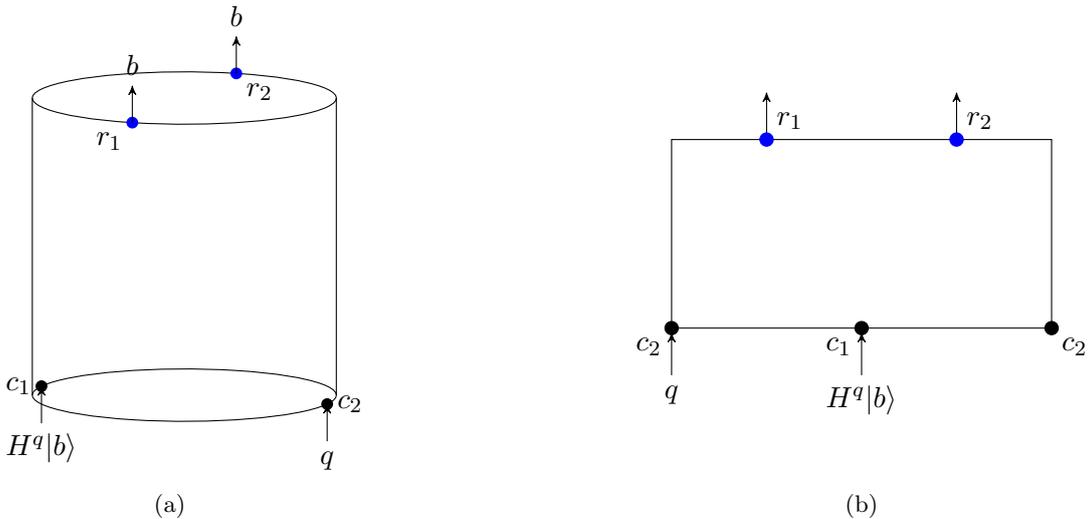

Consider first the bulk picture. If the spacetime region $J_{12\rightarrow 12}$ is nonempty, then it should be possible to successfully complete the $\textbf{B}_{84}$ task in the bulk, by using the local procedure described above. Of course, the details of how one would do this are somewhat complicated. Essentially, one would need to construct a machine within the bulk theory that can (i) carry the qubit $H^q\ket{b}$, (ii) receive and process the classical bit $q$, (iii) measure the qubit, and finally (iv) send the classical signal $b$ to both output points. This machine can be constructed at the asymptotic boundary, at point $c_1$, and then thrown into the scattering region $J_{12\rightarrow 12}$, where it can receive the classical signal $q$ from the input point $c_2$ and then send the output signal $b$ to both $r_1$ and $r_2$.

Can such a machine, or some equivalent method of completing the $\textbf{B}_{84}$ task, actually exist in any given bulk theory? This question is hard to answer definitively: without using our knowledge of the real world, it even seems prohibitively difficult to argue from first principles that such a machine could be constructed out of the standard model. However, since it is believed that string theory is rich enough to describe the real world, we are hopeful that there should exist holographic theories where such objects can be constructed in the bulk.

Let us assume that we are able to find a bulk process that successfully completes the $\textbf{B}_{84}$ task. We now want to understand the same process in the dual boundary perspective. At this point the geometrical observation discussed in the introduction becomes important. Even while the bulk region $J_{12\rightarrow 12}$ is nonempty, the boundary region $\hat{J}_{12\rightarrow 12}$ may be empty. This means that in the CFT picture the task cannot be completed locally. Instead, the local bulk process must be completed nonlocally in the boundary. The necessity of entanglement for any successful nonlocal strategy leads to Theorem \ref{thm:main}. 

Before making this claim more precise, it is useful to consider a generalization of the $\textbf{B}_{84}$ task. In particular we will consider the one-parameter family of tasks denoted $\textbf{B}_{84}^{\times n}$, which consist of the $\textbf{B}_{84}$ task repeated $n$ times in parallel. The parameters $b_i$ and $q_i$ for each repetition of the task are independent. From quantum cryptography we know that the success probability of the $\textbf{B}_{84}^{\times n}$ task is bounded by the following lemma \cite{tomamichel2013monogamy}:
\begin{lemma}\label{lemma:cryptoprobabilitybound} Consider the $\textbf{B}_{84}^{\times n}$ task with the bipartite resource state $\rho_{AB}$, with $A$ available at $c_1$ and $B$ available at $c_2$. Suppose the task must be completed without use of the region $J_{12\rightarrow 12}$. If $I(A:B)=0$ then $p_{\text{suc}}(\textbf{B}_{84}^{\times n}) \leq \beta^n$, with $\beta = \cos^2(\pi/8)$.
\end{lemma}
We refer to the relevant literature for the proof \cite{tomamichel2013monogamy}. Note that the success probability of $\textbf{B}_{84}^{\times n}$ is defined by the rate at which all the outputs $b_i$ are produced correctly. In the above lemma we have considered a general resource state $\rho_{AB}$, which plays the role of the EPR pair in the particular nonlocal protocol we introduced in Figure \ref{fig:nonlocalcircuit}.

Lemma \ref{lemma:cryptoprobabilitybound} gives the strongest possible bound on success probability for the $\textbf{B}_{84}^{\times n}$ task, as the bound is actually achievable. To see this, suppose Alice adopts the following protocol. Upon receiving $H^q\ket{b}$, she measures in the basis intermediate between $\{\ket{0},\ket{1}\}$ and $\{\ket{+},\ket{-}\}$, namely she measures in the basis $\{\ket{\psi_0},\ket{\psi_1}\}$ with 
\begin{align}
    \ket{\psi_0} &= \cos\left( \frac{\pi}{8}\right) \ket{0} + \sin \left( \frac{\pi}{8}\right)\ket{1}, \\ 
    \ket{\psi_1} &= \cos\left(\frac{5\pi}{8}\right) \ket{0} + \sin \left( \frac{5\pi}{8}\right)\ket{1}.
\end{align}
Her guessing strategy is to guess outcome $b=0$ when she measures $\ket{\psi_0}$ and outcome $b=1$ when she measures $\ket{\psi_1}$. Since Alice receives one of the four states $\{\ket{0},\ket{1},\ket{+},\ket{-}\}$ with equal probability, her success probability is
\begin{align}
    p_{\text{suc}} &= \frac{1}{4}|\braket{0}{\psi_0}|^2 + \frac{1}{4}|\braket{1}{\psi_1}|^2 + \frac{1}{4}|\braket{+}{\psi_0}|^2 + \frac{1}{4}|\braket{-}{\psi_1}|^2\nonumber \\
    &= \cos^2\left(\frac{\pi}{8}\right)\nonumber \\
    &=\beta,
\end{align}
with $\beta$ defined as in Lemma \ref{lemma:cryptoprobabilitybound}. Alice repeats this protocol for each parallel repetition of the task, giving her an overall success probability of $\beta^n$.  

Later, we will argue that the $\textbf{B}_{84}$ task can be completed in the boundary with high probability. Since zero mutual information implies low success probability, we may think high success probability implies large mutual information. We prove this in the next lemma. 
\begin{lemma}\label{lemma:mutualInfoBound}
    Consider the $\textbf{B}_{84}^{\times n}$ task when no central region is available. Suppose that the task is completed with probability $p_{\text{suc}} \geq 1 - \epsilon$, using a resource state $\rho_{AB}$. Then $\rho_{AB}$ must satisfy
    \begin{align}
        \frac{1}{2}I(A:B)_{\rho} \geq - \log[2(\epsilon+\beta^n)]
    \end{align}
    with $\beta = \cos^2(\pi/8)$.  
\end{lemma}
\begin{proof}
We begin by recalling that the mutual information can be written as the relative entropy between $\rho_{AB}$ and the product state $\rho_A\otimes \rho_B$, 
\begin{align}
    I(A:B)_{\rho_{AB}} = S(\rho_{AB}||\rho_A\otimes \rho_B).
\end{align}
The relative entropy measures distinguishability of states, so the mutual information is a measure of how distinct $\rho_{AB}$ is from the product state $\rho_A \otimes \rho_B$. In Appendix \ref{appendix:relativeentropyandTrace}, we prove a lower bound on the distinguishability of resource states in terms of their corresponding success probabilities using the trace distance:
\begin{align}
    |p_{\text{suc}}(\rho_{AB}) - p_{\text{suc}}(\rho_{A}\otimes \rho_B)| \leq \frac{1}{2} ||\rho_{AB}-\rho_A\otimes \rho_B||_1.
\end{align}
The trace distance is related to the fidelity by
\begin{align}
    \frac{1}{2}||\rho-\sigma||_1 \leq \sqrt{1-F(\rho,\sigma)},
\end{align}
and the fidelity to the relative entropy by
\begin{align}
    S(\rho||\sigma) \geq - 2 \log F(\rho,\sigma).
\end{align}
Combining these inequalities gives
\begin{align}
    I(A:B) \geq - 2 \log\left[1 - |p_{\text{suc}}(\rho_{AB}) - p_{\text{suc}}(\rho_{A}\otimes \rho_B)|^2\right]
\end{align}
Now using $p_{\text{suc}}(\rho_{AB}) \geq 1 - \epsilon$ and the upper bound $p_{\text{suc}}(\rho_A \otimes \rho_B) \leq \beta^n$ from Lemma \ref{lemma:cryptoprobabilitybound}, we have
\begin{align}
    I(A:B) \geq - 2 \log\left[2(\epsilon + \beta^n)\right],
\end{align}
where we have dropped terms at order $O(\epsilon^2, \epsilon \beta^n, \beta^{2n}).$
\end{proof}

\subsection{The connected wedge theorem}\label{sec:22connectedwedge}

We will argue for Theorem \ref{thm:main}, the connected wedge theorem, by arranging a $\textbf{B}_{84}^{\times n}$ task on an asymptotically AdS$_{d+1}$ spacetime. We will then apply Lemma \ref{lemma:mutualInfoBound} to argue that boundary regions $V_1$ and $V_2$ must have mutual information at order $1/G_N$. To apply Lemma \ref{lemma:mutualInfoBound} to AdS/CFT, we first need to identify the systems $A$ and $B$ appearing in the lemma with subsystems of the CFT. In the context of the nonlocal procedures discussed in the last subsection, the important feature of the $A$ subsystem was that it was available for quantum computations occurring in the future of $c_1$ and the past of $r_1$ and $r_2$. Similarly, the system $B$ was available in the future of $c_2$ and past of $r_1$ and $r_2$. This indicates we should identify $A$ with $V_1$ and $B$ with $V_2$. 

\begin{figure}
    \centering
    \begin{tikzpicture}[scale=0.75]
    
    \draw[thick, black, fill=black!60!,opacity=0.8] (-6,0) -- (-4,2) -- (-2,0) -- (-4,-2) -- (-6,0);
    
    \draw[thick, black, fill=black!60!,opacity=0.8] (6,0) -- (4,2) -- (2,0) -- (4,-2) -- (6,0);
    
    \draw[lightgray] (-10,8) -- (10,8) -- (10,-2) -- (-10,-2) -- (-10,8);
    
    \draw[dashed] (-6,0) -- (0,6);
    \draw[dashed] (6,0) -- (0,6);
    \draw[dashed] (-2,0) -- (-10,8);
    \draw[dashed] (2,0) -- (10,8);

    \draw[->] (0,6) --(0,7);
    \draw[->] (10,8) -- (10,9);
    \draw[->] (-10,8) -- (-10,9);
    
    \node[left] at (-0,7) {$b$};
    \node[left] at (10,9) {$b$};
    \node[right] at (-10,9) {$b$};
    
    \draw[fill=black] (-4,-2) circle (0.15);
    \draw[fill=black] (4,-2) circle (0.15);
    \node[below left] at (-4,-2) {$c_1$};
    \node[below right] at (4,-2) {$c_2$};
    
    \draw[->] (-4,-3) -- (-4,-2.18);
    \draw[->] (4,-3) -- (4,-2.18);
    \node[right] at (-4,-3) {$H^q\ket{b}$};
    \node[left] at (4,-3) {$q$};
    
    \draw[thick,fill=lightgray] (-2,0) -- (0,2) -- (2,0) -- (0,-2) -- (-2,0);
    \draw[thick,fill=lightgray] (-6,0) -- (-10,4) -- (-10,-2) -- (-8,-2) -- (-6,0);
    \draw[thick,fill=lightgray] (6,0) -- (10,4) -- (10,-2) -- (8,-2) -- (6,0);
    
    \node at (4,0) {$V_2$};
    \node at (-4,0) {$V_1$};

    \node at (0,0) {$X_1$};
    \node at (-8,0) {$X_2$};
    \node at (8,0) {$X_2$};
    
    \draw[fill=blue] (10,8) circle (0.15);
    \draw[fill=blue] (-10,8) circle (0.15);
    \draw[fill=blue] (0,6) circle (0.15);
    \node[above left] at (-10,8) {$r_2$};
    \node[above right] at (10,8) {$r_2$};
    \node[above right] at (0,6) {$r_1$}; 
    
    \end{tikzpicture}
    \caption{Boundary view of the $\textbf{B}_{84}$ task. Input regions $V_1$ and $V_2$ are shown in dark gray. The regions $X_i$, shown in light gray, sit between the input regions and add some complications to the boundary proof of Theorem \ref{thm:main}.}
    \label{fig:BB84taskWithRegions}
\end{figure}
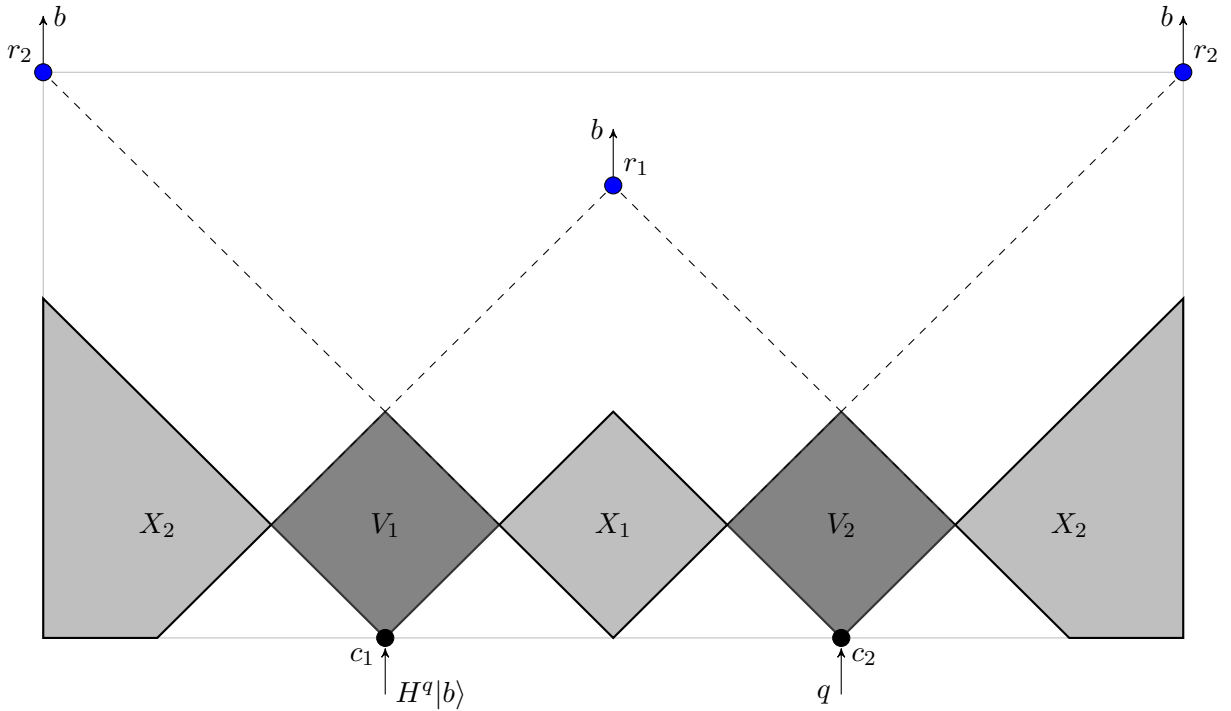
 
There is a complication however in applying Lemma \ref{lemma:mutualInfoBound} to the CFT. This is because in addition to the input region systems $V_1V_2$, there is also a system $X$ that sits between the input regions. Potentially, this can act as an additional resource with which to complete the task. To be precise, we define the complementary domain of dependence $X$ by\footnote{Here the overline denotes set closure.}
\begin{align} \label{eq:complementary-DOD}
    X = \bar{[\hat{J}^+(V_1 \cup V_2) \cup \hat{J}^-(V_1 \cup V_2)]^c}.
\end{align}
In $(2+1)$ bulk dimensions, this region splits into two disconnected components $X_1$ and $X_2$ as seen in Figure \ref{fig:BB84taskWithRegions}. This issue was addressed incorrectly in the argument for the connected wedge theorem appearing in \cite{may2019quantum}, and the existence of the region $X$ remains the main obstruction to providing a complete proof of Theorem \ref{thm:main} using quantum information constraints on the boundary. 

In fact, in appendix \ref{sec:22appendix}, we provide an example of a protocol that exploits the region $X$ to complete the $\textbf{B}_{84}$ task nonlocally with zero mutual information between the input regions. Effectively, entanglement between the input regions is `hidden' by applying a random choice of unitary matrix to the qubit in region $V_2$, with the classical information describing which unitary was chosen stored in both $X_1$ and $X_2$. The application of a random unitary destroys the entanglement between $V_1$ and $V_2$, but the `hidden' entanglement in $X_1$ and $X_2$ is still accessible by the output points.

We do not believe that this loophole is relevant in the context of holography. The large amount of classical correlation between $X_1$ and $X_2$ needed to replace quantum correlations between $V_1$ and $V_2$ requires, in turn, a large amount of GHZ-type entanglement between the three regions $X_1$, $X_2$ and $V_1 \cup V_2$. Such patterns of correlation are known not to exist in stabilizer tensor networks \cite{Nezami:2016zni}, and we do not expect them to exist in holography. Instead, when the entanglement wedge is disconnected, the leading-order four-party entanglement structure is expected to be a combination of: (a) (divergent) bipartite entanglement between neighboring regions, (b) bipartite entanglement between $X_1$ and $X_2$, and (c) four-party perfect-tensor entanglement \cite{Cui:2018dyq}. The first two possibilities cannot help, since any system in $X_i$ may be moved forward in time into the future light cones of the $V_i$ without lowering the maximum success probability of the protocol. This means entanglement between the $X_i$ can be replaced with entanglement created locally within one of the $V_i$ and then distributed, and similarly for entanglement between $V_i$ and $X_i$. We believe that the third also cannot help, because there is no information accessible to both $X_1$ and $X_2$ that can be used to extract entanglement between $V_1$ and $V_2$. As far as we know, however, this statement has not been proved; as such, four-party perfect tensor entanglement therefore cannot be completely ruled out as a resource for the boundary theory to complete the $\mathbf{B}_{84}$ task without large mutual information between $V_1$ and $V_2$.\footnote{We stress that this problem only affects the quantum information argument of this section; the bulk proof of Theorem \ref{thm:main} that we give in Section \ref{sec:GR} is unaffected}.

Finally we are ready to give the quantum information argument for the connected wedge theorem from the perspective of quantum tasks. We restate the theorem here.
\vspace{0.3cm}\\
\noindent \textbf{Theorem} \ref{thm:main} \emph{
    Let $\{c_1, c_2, r_1, r_2\}$ be a bulk-only scattering configuration on the boundary of an asymptotically AdS spacetime with a holographic dual. Let $V_1$ and $V_2$ be boundary regions defined by
    \begin{equation}
        V_1 = \hat{J}_{1 \rightarrow 12} \quad \text{and} \quad V_2 = \hat{J}_{2 \rightarrow 12}.
    \end{equation}
    Then $V_1$ and $V_2$ are spacelike-separated domains of dependence, and the entanglement wedge of $V_1 \cup V_2$ is connected.
}
\vspace{0.3cm}\\
\begin{argument} Since $J_{12\rightarrow 12}$ is nonempty, the $\textbf{B}_{84}^{\times n}$ task can be completed in the bulk with high success probability, say $p_{\text{suc}}\geq1-\epsilon$. We should be careful though to specify exactly how large $n$ can be taken, and how small $\epsilon$ is. In particular we will be interested in how $\epsilon$ and $n$ behave with $G_N$. We claim that in the limit $G_N\rightarrow 0$, we may take
\begin{align}
    \epsilon &\rightarrow 0, \nonumber \\
    n &\rightarrow \infty.
\end{align}
To see $\epsilon \rightarrow 0$, consider that on a completely classical and fixed background we expect that the bulk computation, applying a Hadamard $H$ and measuring in the computational basis, can be completed perfectly. Since we obtain a classical, fixed, background when $G_N\rightarrow 0$, we expect $\epsilon \rightarrow 0$ in this limit. To argue $n\rightarrow \infty$, note that the constraint on $n$ is that we must avoid sending so many qubits into the bulk that we backreact and close off the region $J_{12\rightarrow 12}$. This requires
\begin{align}
    n(\Delta E) < O(1/G_N),
\end{align}
where $\Delta E$ is the energy carried by each qubit, since in Einstein's equations the energy $n (\Delta E)$ couples to the geometry with a factor of $G_N$. This shows we can take $n\rightarrow \infty$ so long as $n$ grows more slowly than $1/G_N$.

From the bulk perspective we have that the $\mathbf{B}_{84}^{\times n}$ task can be completed with $p_{\text{suc}}\geq 1-\epsilon$, where $\epsilon \rightarrow 0$ and $n\rightarrow \infty$ as $G_N\rightarrow 0$. The AdS/CFT correspondence then tells us that the same must be true of the boundary task. Under the assumptions explained above, we may apply Lemma \ref{lemma:mutualInfoBound} to obtain
\begin{align}\label{eq:prooflowerbound}
    \frac{1}{2}I(V_1:V_2) \geq -\log[2(\epsilon + \beta^n)].
\end{align}
Since $\epsilon\rightarrow 0$ and $n\rightarrow \infty$ as $G_N\rightarrow 0$, the right-hand side grows as $G_N\rightarrow 0$. Consequently the mutual information $I(V_1:V_2)$ is strictly larger than $O(1)$. Since the HRRT formula \eqref{eq:HRRT} implies that the mutual information can only have an $O(1/G_N)$ piece and an $O(1)$ piece, the mutual information must be $O(1/G_N)$ and so the entanglement wedge must be connected.
\end{argument}

In this argument, we appeal to the HRRT formula to show that $I(V_1 : V_2)$ growing faster than $O(1)$ in the $G_N \rightarrow 0$ limit implies it must be $O(1/G_N).$ It is interesting to understand what would be required of the tasks argument to show that $I(V_1:V_2)=O(1/G_N)$ without using the HRRT formula. It would suffice to take $\epsilon = O(e^{-1/G_N})$, since also we can take $n=O(1/G_N^\alpha)$ for $\alpha\in [0,1)$ without backreacting. Then the lower bound \eqref{eq:prooflowerbound} becomes
\begin{align}
    I(V_1:V_2) \geq O(1/G_N^\alpha),
\end{align}
for any $\alpha \in [0,1)$, from which we can conclude $I(V_1:V_2)=O(1/G_N)$.\footnote{In \cite{may2019quantum}, where the connected wedge theorem was proved initially, $\epsilon=0$ was implicitly assumed. As we see here, that proof is not robust to allowing errors larger than $O(e^{-1/G_N})$, so the argument there actually requires a strong assumption on the error.} We see that $\epsilon=O(e^{-1/G_N})$ is strong enough to reproduce the conclusion from the HRRT formula that the mutual information is order $1/G_N$. While $\epsilon=O(e^{-1/G_N})$ is small, it is exactly the size of correction that one would expect from nonperturbative instanton effects associated to subleading bulk saddles. For example, the leading errors in entanglement wedge reconstruction come from subleading bulk saddle points and scale in exactly this way.\footnote{For information-theoretic arguments that suggest that this should generally be true, see \cite{Hayden:2018khn, Akers:2019wxj}. For explicit calculations of these errors in a very simple toy model, see \cite{Penington:2019kki}.}

Note that the proof techniques of this section are all time-reversal invariant. One could equally well formulate a ``past-directed'' $\mathbf{B}_{84}$ task, where one is given the inputs at $r_1$ and $r_2$, allowed to propagate signals backwards through time, and tasked with returning the outputs at $c_1$ and $c_2$ This implies that the output regions $\hat{J}_{1 \rightarrow 12}$ and $\hat{J}_{2 \rightarrow 12}$ have a connected entanglement wedge as well.

\section{Classical proof from the holographic dictionary}
\label{sec:GR}

In this section, we provide a proof of Theorem \ref{thm:main} in the bulk using classical general relativity. Our main tools will be (i) the HRRT formula \eqref{eq:HRRT}, (ii) the focusing theorem, and (iii) the maximin prescription for HRRT surfaces \cite{maximin}. While the focusing theorem holds only for spacetimes satisfying the null energy condition, all proofs in this section can be generalized to semiclassical spacetimes with quantum matter by using generalized entropy instead of area and quantum extremal surfaces \cite{QES} instead of extremal surfaces in the HRRT formula, and assuming both the quantum focusing conjecture \cite{QFC} and quantum maximin \cite{quantum-maximin}.

The usual pedagogy of the focusing theorem first defines the expansion of a null sheet as the fractional rate of change in the area of its cross-sections, then shows that this expansion is non-positive on particular null sheets of interest. In subsection \ref{sec:area-theorems}, we present a different framing of the focusing theorem, which places Stokes' theorem and the geometry of the null sheet at the center of the discussion. This perspective is useful for understanding Theorem \ref{thm:main}, as it allows us to see explicitly how the extremal surfaces $\mathcal{E}_{V_1}^{\text{min}}$ and $\mathcal{E}_{V_2}^{\text{min}}$ interface with the bulk scattering region $J_{12\rightarrow 12}$; the extremal surfaces bound a null membrane that extends deep into the bulk, allowing the extremal surfaces to probe the scattering region through Stokes' theorem. The details of this null membrane construction and the geometric proof of Theorem \ref{thm:main} are provided in subsection \ref{sec:null-membrane}.

We assume throughout this section that (i) the bulk spacetime satisfies the null energy condition, and (ii) the HRRT surface $\mathcal{E}_{V_1 \cup V_2}^{\text{min}}$ can be found by a maximin procedure.\footnote{In \cite{maximin}, it was proven that the maximin surface is equivalent to the HRRT surface whenever the maximin surface exists. It was argued there and in \cite{maximin2} that maximin surfaces exist for generic boundary domains of dependence in a large class of spacetimes.} We will also assume (iii) that the spacetime is \emph{AdS-hyperbolic}, i.e., that the ``unphysical spacetime'' consisting of the spacetime together with its conformal boundary admits a Cauchy slice.\footnote{A spacetime having a Cauchy slice is equivalent to the property that for any two spacetime points $p$ and $q$, the causal region $J^+(p) \cap J^-(q)$ is compact \cite{gerochDOD}. This essentially requires that there are no ``holes'' or secret asymptotic regions in the bulk; in the case of a spacetime with multiple asymptotically AdS regions, such as a multiboundary wormhole, the spacetime together with all conformal boundaries is still globally hyperbolic.} As in \cite{maximin}, we will also assume that spacelike slices of any conformal compactification are themselves compact. We work in an arbitrary number of bulk dimensions in subsection \ref{sec:area-theorems}, and specialize to $(2+1)$ dimensions in subsection \ref{sec:null-membrane} for ease of visualization; however, the proof techniques of subsection \ref{sec:null-membrane} are not dimension-dependent.

\subsection{Area theorems for null surfaces}
\label{sec:area-theorems}

Let $\mathcal{N}$ be a codimension-$1$ null surface in a $(d+1)$-dimensional spacetime. Further assume that $\del \mathcal{N}$ consists of two spacelike surfaces, $\Sigma_0$ and $\Sigma_1$, such that $\Sigma_1$ is in the future of $\Sigma_0$. Each of these surfaces is codimension-$2$, and has a $(d-1)$-index volume form $\bm{\tilde{\epsilon}}_{a_1 \dots a_{d-1}}.$ The difference in area between the two surfaces is given by
\begin{equation}
    \text{area}(\Sigma_2) - \text{area}(\Sigma_1)
        = \int_{\Sigma_2} \bm{\tilde{\epsilon}}_{a_1 \dots a_{d-1}}
            - \int_{\Sigma_1} \bm{\tilde{\epsilon}}_{a_1 \dots a_{d-1}},
\end{equation}
which by Stokes' theorem may be computed as
\begin{equation} \label{eq:stokes-interior}
    \text{area}(\Sigma_2) - \text{area}(\Sigma_1)
        = \int_\mathcal{N} (\dd\bm{\tilde{\epsilon}})_{a_1 \dots a_{d}}
\end{equation}
for any smooth extension of the boundary volume form $\bm{\tilde{\epsilon}}_{a_1 \dots a_{d-1}}$ to the interior of $\mathcal{N}$. Suppose further that $\mathcal{N}$ is generated by null geodesics, i.e., that each point $p \in \mathcal{N}$ lies on some future-directed null geodesic $\gamma$ that lies entirely within $\mathcal{N}$ between $\Sigma_0$ and $p$. Choosing an affine parameter $\lambda$ for these null geodesics gives (i) an affinely parametrized tangent vector $k^a$ satisfying $k^b \Del_b k^a = 0$, and (ii) a preferred spacelike foliation of $\mathcal{N}$ given by the surfaces of constant $\lambda$. If this affine parameter is chosen such that the future boundary $\Sigma_{1}$ is the surface $\lambda = 1$, then the corresponding spacelike foliation of $\mathcal{N}$ gives a preferred smooth extension of $\bm{\tilde{\epsilon}}_{a_1 \dots a_{d-1}}$ to the null surface $\mathcal{N}$ by taking $\bm{\tilde{\epsilon}}_{a_1 \dots a_{d-1}}$ on each surface of constant $\lambda$ to be the volume form on that surface. We will take this to be the definition of the $(d-1)$-form $\bm{\tilde{\epsilon}}_{a_1 \dots a_{d-1}}$ that appears in equation \eqref{eq:stokes-interior}. Our goal is now to show that $(\dd\bm{\tilde{\epsilon}})_{a_1 \dots a_{d}}$ is proportional to the expansion of $\mathcal{N}$ with respect to $\lambda$, thus bringing equation \eqref{eq:stokes-interior} in contact with the familiar formulation of the focusing theorem.

Let $\bm{\epsilon}_{a_1 \dots a_{d}}$ be a volume form for the null surface $\mathcal{N}.$ There is no uniquely preferred volume form on $\mathcal{N}$, since the metric on null surfaces is degenerate, but we may specify $\bm{\epsilon}_{a_1 \dots a_{d}}$ by noting that $\bm{\epsilon}_{a_1 \dots a_d}$ is a top-level form\footnote{By ``top-level form'', we mean a $k$-form on a surface with a $k$-dimensional tangent space. On any surface, the space of top-level forms is one-dimensional.} on $\mathcal{N}$ and hence unique up to multiplication by a smooth function. We fix that smooth function by requiring $\bm{\epsilon}_{a_1 \dots a_d}$ to satisfy
\begin{equation} \label{eq:null-normalization}
    k^{a_1} \bm{\tilde{\epsilon}}^{a_2 \dots a_d} \bm{\epsilon}_{a_1 \dots a_d} = \bm{\tilde{\epsilon}}^{a_2 \dots a_d} \bm{\tilde{\epsilon}}_{a_2 \dots a_d} = (d-1)!.
\end{equation}
Note that $\bm{\epsilon}_{a_1 \dots a_d}$ and $(\dd\bm{\tilde{\epsilon}})_{a_1 \dots a_{d}}$ are both top-level forms on $\mathcal{N}$, so they must satisfy an equation of the form
\begin{equation}
    (\dd\bm{\tilde{\epsilon}})_{a_1 \dots a_{d}}
        = \theta \bm{\epsilon}_{a_1 \dots a_d}.
\end{equation}
The function $\theta$ can be computed using equation \eqref{eq:null-normalization} as
\begin{equation} \label{eq:alpha-prelim}
    \theta (d-1)!
        = k^{a_1} \bm{\tilde{\epsilon}}^{a_2 \dots a_d} (\dd\bm{\tilde{\epsilon}})_{a_1 \dots a_{d}}.
\end{equation}
The right-hand side of this expression can be computed by noting that $\bm{\tilde{\epsilon}}$ is given by $\iota_{k} \bm{\epsilon}$, i.e., by contracting $k^a$ into the first index of $\bm{\epsilon}_{a_1 \dots a_d}$. This follows from the relationship between $\bm{\tilde{\epsilon}}_{a_1 \dots a_{d-1}}$ and $\bm{\epsilon}_{a_1 \dots a_d}$ given in equation \eqref{eq:null-normalization}, together with the fact that $\bm{\tilde{\epsilon}}_{a_1 \dots a_{d-1}}$ is a top-level form on the spacelike cross-sections of $\mathcal{N}$. This observation, together with the Lie derivative formula
\begin{equation}
    \mathscr{L}_{k} (\bm{\epsilon}) = \iota_k \dd(\bm{\epsilon}) + \dd(\iota_k \bm{\epsilon}),
\end{equation}
allows us to rewrite equation \eqref{eq:alpha-prelim} as 
\begin{equation} \label{eq:alpha-mid}
    \theta (d-1)!
        = k^{a_1} \bm{\tilde{\epsilon}}^{a_2 \dots a_d} \left[ \mathscr{L}_{k}(\bm{\epsilon}_{a_1 \dots a_d}) - k^b \dd(\bm{\epsilon})_{b a_1 \dots a_{d}} \right] = k^{a_1} \bm{\tilde{\epsilon}}^{a_2 \dots a_d} \mathscr{L}_{k}(\bm{\epsilon}_{a_1 \dots a_d}).
\end{equation}
The second equality follows from the fact that $\bm{\epsilon}$ is a top-level form and hence closed on $\mathcal{N}.$

The Lie derivative appearing in equation \eqref{eq:alpha-mid} can be expanded explicitly as 
\begin{equation}
    \mathscr{L}_{k}(\bm{\epsilon}_{a_1 \dots a_d})
        = k^b \Del_b \bm{\epsilon}_{a_1 \dots a_d}
            + \bm{\epsilon}_{b \dots a_d} \Del_{a_1} k^b
            + \bm{\epsilon}_{a_1 b \dots a_d} \Del_{a_2} k^b
            + \dots
            + \bm{\epsilon}_{a_1 \dots b} \Del_{a_d} k^b.
\end{equation}
Plugging this back into equation \eqref{eq:alpha-mid}, one finds that (i) the first term vanishes due to the relations $k^b \Del_b k^a = 0$ and $\bm{\tilde{\epsilon}}^{a_2 \dots a_d} \bm{\tilde{\epsilon}}_{a_2 \dots a_d} = (d-1)!$, and (ii) the second term vanishes due to the relation $k^b \Del_b k^a = 0$. We are left with the equation
\begin{equation}
    \theta (d-1)!
        = k^{a_1} \bm{\tilde{\epsilon}}^{a_2 \dots a_d} \left( \bm{\epsilon}_{a_1 b \dots a_d} \Del_{a_2} k^b
            + \dots
            + \bm{\epsilon}_{a_1 \dots b} \Del_{a_d} k^b \right).
\end{equation}
Each term in this equation can be computed using the relations
\begin{equation}
    k^{a_1} \bm{\epsilon}_{a_1 \dots a_d} = \bm{\tilde{\epsilon}}_{a_2 \dots a_d} \qquad \text{and} \qquad \bm{\tilde{\epsilon}}^{b a_{3} \dots a_d} \bm{\tilde{\epsilon}}_{c a_3 \dots a_d} = h^{b}{}_{c} (d-2)!,
\end{equation}
where $h_{bc}$ is the projection onto the leaves of the foliation, i.e., the pullback of the metric to the spacelike cross-sections of $\mathcal{N}$. The resulting expression is
\begin{equation} \label{eq:expansion}
    \theta
        = h^{ab} \Del_{a} k_b,
\end{equation}
which is the usual definition of the expansion of $\mathcal{N}$ with respect to the generator $k^a$. The final result of this calculation is that the difference in area between two spacelike slices of a null surface is given by the integral of the expansion over the interior, i.e., by
\begin{equation} \label{eq:stokes-area}
    \text{area}(\Sigma_2) - \text{area}(\Sigma_1)
        = \int_{\mathcal{N}} \theta \bm{\epsilon}_{a_1 \dots a_{d}}
        = \int_{\mathcal{N}} (h^{ab} \Del_{a} k_b) \bm{\epsilon}_{a_1 \dots a_{d}}.
\end{equation}

The expansion $\theta$ of a congruence of affinely parametrized geodesics is governed by the Raychaudhuri equation \cite{raychaudhuri},
\begin{equation} \label{eq:raychaudhuri}
    \frac{d\theta}{d\lambda} = - \frac{\theta^2}{d-1} - \sigma_{ab} \sigma^{ab} + \omega_{ab} \omega^{ab} - R_{ab} k^a k^b.
\end{equation}
Here the \emph{shear} $\sigma_{ab}$ and the \emph{twist} $\omega_{ab}$ are defined as the symmetric traceless and antisymmetric parts of the tensor
\begin{equation}
    B_{ab} = h_{a}{}^{d} h_{b}{}^{c} \Del_{c} k_d,
\end{equation}
respectively. In conjunction with equation \eqref{eq:stokes-area}, the Raychaudhuri equation constrains the difference in area between two cross-sections of a null surface generated by geodesics.

There are two kinds of null surfaces whose expansions are particularly constrained by the Raychaudhuri equation: lightsheets of extremal surfaces, and causal horizons. The first of these is defined as the boundary of the future (or past) of an extremal surface $\Sigma$: $\mathcal{N}^{\pm} = \del J^{\pm}(\Sigma).$\footnote{Surfaces of this type have been studied by many authors, and were given the name ``lightsheets'' in \cite{boussobound}.} Such a surface is generated by null geodesics that leave $\Sigma$ orthogonally;\footnote{This is guaranteed by the AdS-hyperbolicity of the spacetime.} with respect to any affine parametrization of these geodesics, $\mathcal{N}^{\pm}$ has everywhere non-positive expansion. This follows from the following facts: (i) the twist $\omega_{ab}$ vanishes on $\mathcal{N}^{\pm}$ since the geodesics are hypersurface-orthogonal,\footnote{See Chapter 9 of \cite{waldbook} for a review of this statement.} (ii) the term $(- R_{ab} k^a k^b)$ is non-positive by the null energy condition, (iii) the expansion vanishes on $\Sigma$ by its extremality, and (iv) any geodesic whose expansion passes through $-\infty$ leaves $\del J^+(\Sigma)$.\footnote{A nice discussion of this property of the boundary of the future of a set appears in Section 5 of \cite{wittensingularities}.} Conditions (i) and (ii) together imply the equation $d \theta / d \lambda \leq 0$; this does not assume the extremality of $\Sigma$, and is what is often called the \emph{focusing theorem}. Conditions (iii) and (iv) then imply that the expansion of $\mathcal{N}^{\pm}$ is non-positive, since the expansion can only become positive by passing through $\theta = - \infty$, at which point the relevant geodesic leaves $\mathcal{N}^{\pm}$.

The other kind of null surface relevant to the proof of Theorem \ref{thm:main} is the \emph{causal horizon}, i.e., the boundary of the past of a point at infinity. For concreteness, let $S$ be a set of points at the spacetime boundary. Then $\del J^{-}(S)$ is generated by past-directed null geodesics that go to the spacetime boundary in the affine parameter limit $\lambda \rightarrow -\infty$. The expansion of these geodesics with respect to any affine parameter must be non-positive, for if the expansion is positive at some point $p$ with affine parameter $\lambda_p$, then the Raychaudhuri equation \eqref{eq:raychaudhuri} implies that the expansion must have been $+\infty$ at some point between $\lambda_p$ and the ``earlier'' affine parameter
\begin{equation}
    \lambda = \lambda_p - \frac{d-1}{\theta_p}.
\end{equation}
This contradicts the assumption that $p$ is in the boundary of the past of $S$ by point (iv) of the preceding paragraph. We conclude that the expansion of any causal horizon must be non-positive with respect to past-directed generators. The equivalent statement that the expansion of a causal horizon must be non-negative with respect to future-directed generators is Hawking's area theorem.

That the expansion is non-positive for future-directed extremal lightsheets and past-directed causal horizons, and that the area difference between cross-sections of a null surface equals the surface integral of the expansion, are the critical ingredients in proving Theorem \ref{thm:main}. We provide the proof in the following subsection.

\subsection{Proof via the null membrane}
\label{sec:null-membrane}

Suppose now that $(\mathcal{M}, g_{ab})$ is an asymptotically AdS$_{2+1}$ spacetime with a bulk-only scattering configuration $\{c_1, c_2, r_1, r_2\}$. Let $V_1 = \hat{J}_{1 \rightarrow 12}$ and $V_2 = \hat{J}_{2 \rightarrow 12}$ be the boundary domains of dependence defined in Theorem \ref{thm:main}. Let $\mathcal{E}_{V_1}^{\text{min}}$ and $\mathcal{E}_{V_2}^{\text{min}}$ be the HRRT surfaces of $V_1$ and $V_2$, respectively. Our goal is to show that the joint HRRT surface $\mathcal{E}_{V_1 \cup V_2}^{\text{min}}$ cannot equal $\mathcal{E}_{V_1}^{\text{min}} \cup \mathcal{E}_{V_2}^{\text{min}}.$ This implies that the entanglement wedge is connected, and hence proves Theorem \ref{thm:main}. We proceed by constructing, on any complete achronal slice containing $\mathcal{E}_{V_1}^{\text{min}} \cup \mathcal{E}_{V_2}^{\text{min}}$, a topologically distinct surface whose area is smaller than that of $\mathcal{E}_{V_1}^{\text{min}} \cup \mathcal{E}_{V_2}^{\text{min}}.$ This contradicts the maximin procedure \cite{maximin}, which implies that the true HRRT surface of $V_1 \cup V_2$ must be globally minimal on some complete achronal slice.

We begin by defining, for any complete achronal slice $\Sigma$ containing $\mathcal{E}_{V_1}^{\text{min}} \cup \mathcal{E}_{V_2}^{\text{min}}$, the \emph{null membrane} $\mathcal{N}_{\Sigma}$. $\mathcal{N}_{\Sigma}$ is a codimension-$1$ null surface in the bulk constructed from two distinct pieces. The first piece, which we call the \emph{lift}\footnote{The terms \emph{lift}, \emph{ridge}, and \emph{slope} used to define the null membrane are borrowed from skiing terminology for reasons that we hope will become obvious.}, is given by
\begin{equation} \label{eq:lift}
    \mathcal{L} = \del J_{\text{in}}^+(\mathcal{E}_{V_1}^{\text{min}} \cup \mathcal{E}_{V_2}^{\text{min}}) \cap J^{-}(r_1) \cap J^{-}(r_2).
\end{equation}
This surface consists of null geodesics on the boundary of the future of $\mathcal{E}_{V_1}^{\text{min}} \cup \mathcal{E}_{V_2}^{\text{min}}$ that lie in the past of both $r_1$ and $r_2$. Note that we restrict to those null geodesics in the ``inward-pointing'' direction, i.e., those that move initially away from the boundary domains of dependence. The lift $\mathcal{L}$ consists of portions of the lightsheets of $\mathcal{E}_{V_1}^{\text{min}}$ and $\mathcal{E}_{V_2}^{\text{min}}$ up until the points where those lightsheets meet, at which point $\mathcal{L}$ caps off with a spacelike  \emph{ridge}, $\mathcal{R}$. A surface of this kind is sketched in Figure \ref{fig:lift}. Note that the lift is independent of $\Sigma$.

\begin{figure}
    \centering
    \subfloat[\label{fig:lift}]{
    \includegraphics[scale=1.1]{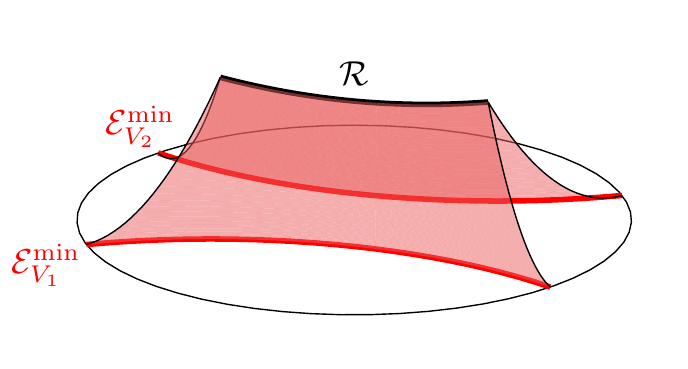}
    }
    \subfloat[\label{fig:slope}]{
    \includegraphics[scale=1.1]{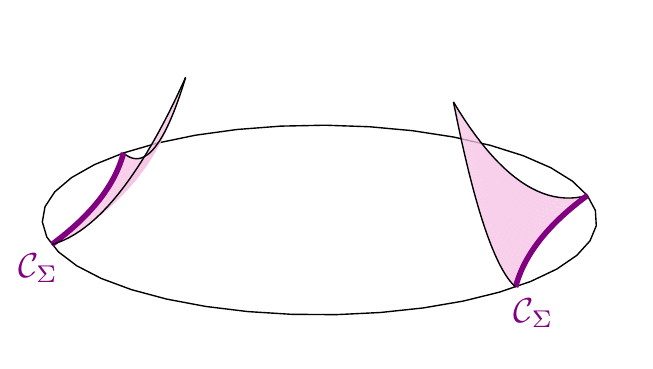}
    }
    \caption{(a) The lift $\mathcal{L}$ generated by two HRRT surfaces $\mathcal{E}_{V_1}^{\text{min}}$ and $\mathcal{E}_{V_2}^{\text{min}}$ (cf. eq.\eqref{eq:lift}). The HRRT surfaces and the ridge $\mathcal{R}$ are marked with heavy curves. The boundary circle represents a complete achronal slice $\Sigma$ containing the HRRT surfaces, though the definition of the lift is $\Sigma$-independent. (b) The slope $\mathcal{S}_{\Sigma}$ terminating on a complete achronal slice $\Sigma$ containing the HRRT surfaces. The contradiction surface $\mathcal{C}_{\Sigma}$ is marked with heavy curves. (Cf. eq.\eqref{eq:slope}.)}
    \label{fig:lift-slope}
\end{figure}

The second piece of the null membrane, which we call the \emph{slope}, is defined by
\begin{equation} \label{eq:slope}
    \mathcal{S}_{\Sigma} = \del[J^-(r_1) \cap J^-(r_2)] \cap J^-[\del J_{\text{in}}^+(\mathcal{E}_{V_1}^{\text{min}}\cup\mathcal{E}_{V_2}^{\text{min}})] \cap J^+(\Sigma).
\end{equation}
This surface consists of all those points on the past lightsheets of $r_1$ and $r_2$ that (i) lie in the past of the points where the lightsheets meet, (ii) lie in the past of the lightsheets of the minimal surfaces of $V_1$ and $V_2$, and (iii) lie in the future of $\Sigma$. A surface of this type is sketched in Figure \ref{fig:slope}. The null membrane is defined to be the union of these two surfaces, $\mathcal{N}_{\Sigma} = \mathcal{L} \cup \mathcal{S}_{\Sigma}$. For a bulk-only scattering configuration, the null membrane generically has a topological structure like that shown in Figure \ref{fig:null-membrane}.

\begin{figure}
    \centering
    \hspace{-0.5cm}
    \subfloat[\label{fig:null-membrane}]{
    \raisebox{30pt}{\includegraphics[scale=0.9]{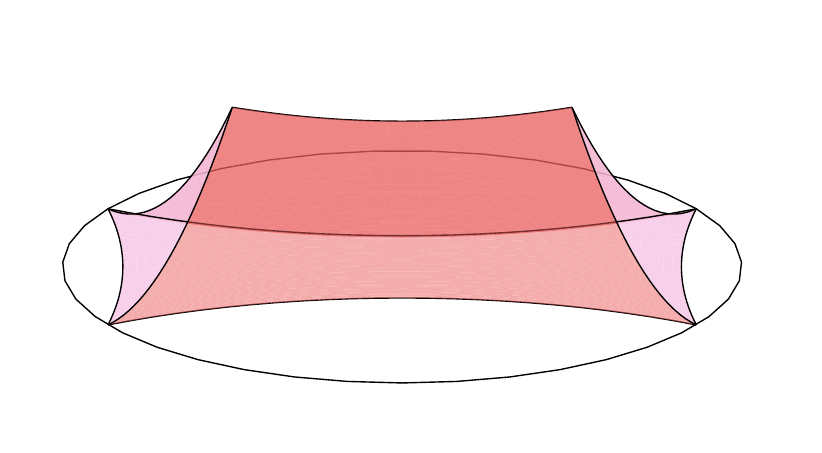}}
    }
    \hspace{10pt}
    \subfloat[\label{fig:membrane-CS}]{
    \includegraphics[scale=0.9]{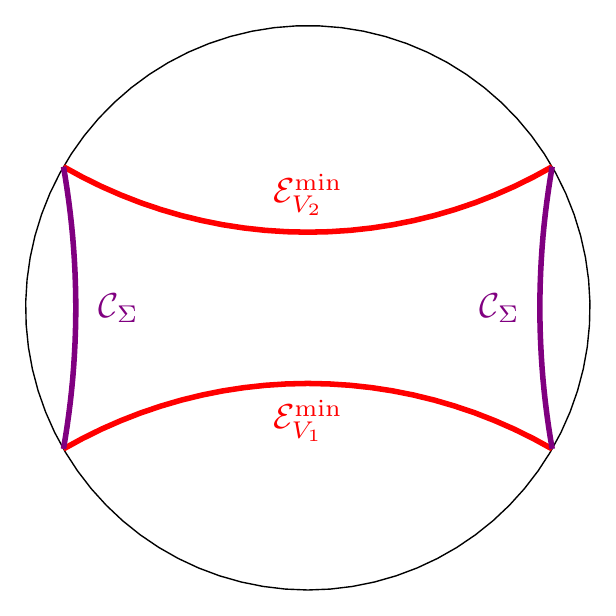}
    }
    \caption{(a) An example of a null membrane $\mathcal{N}_{\Sigma}$ for a bulk-only scattering configuration. The null membrane generically has a spacelike ridge as its future boundary and four spacelike seams where the lift and slope meet; more on this in Figure \ref{fig:cusps}. (b) The past boundary of $\mathcal{N}_{\Sigma}$ is topologically $S^1$ on any complete achronal slice containing the HRRT surfaces of $V_1$ and $V_2$.}
    \label{fig:membrane-zoom}
\end{figure}

The null membrane has several important properties. First, and most importantly, its boundary on $\Sigma$ is nonempty and topologically $S^1$ --- see Figure \ref{fig:membrane-CS} for an example. That the intersections of $\del J^{-}(r_1)$ and $\del J^{-}(r_2)$ with $\Sigma$ must be made up of curves terminating on the spacelike boundaries of $V_1$ and $V_2$ follows from the fact that the spacelike boundaries of $V_1$ and $V_2$ are \emph{defined} to be points on the past lightsheets of $r_1$ and $r_2.$ That these curves must be connected follows from AdS-hyperbolicity. That these curves must exist follows from the fact that the past lightsheets of $r_1$ and $r_2$ on the boundary are null surfaces; the only way these curves can fail to exist is if the lightsheets have already hit the spacetime boundary and disappeared, but this cannot happen on any achronal slice containing the spacelike boundary points of $V_1$ and $V_2$. We label the past boundary of $\mathcal{S}_{\Sigma}$ the \emph{contradiction surface} $\mathcal{C}_{\Sigma}.$

The other important property of the null membrane is that its cusps have a simple structure. This structure is sketched in Figure \ref{fig:cusps}. As mentioned previously, the lift $\mathcal{L}$ has a simple spacelike ridge $\mathcal{R}$ where the lightsheets of $\mathcal{E}_{V_1}^{\text{min}}$ and $\mathcal{E}_{V_2}^{\text{min}}$ meet. That this ridge is nonempty follows from the assumption that the bulk scattering region $J_{12 \rightarrow 12}$ is nonempty. That assumption implies that the lightsheets of $V_1$ and $V_2$ meet in the past of $r_1$ and $r_2$. This in turn implies that the lightsheets of $\mathcal{E}_{V_1}^{\text{min}}$ and $\mathcal{E}_{V_2}^{\text{min}}$ meet in the past of $r_1$ and $r_2$ by the property that the HRRT surface of $V_i$ must always lie spacelike outside the spacelike boundary of its causal wedge \cite{maximin}. Furthermore, the lift $\mathcal{L}$ and slope $\mathcal{S}_{\Sigma}$ meet at a set of four spacelike cusps, which we label $\mathcal{A}$. That these cusps exist follows from the assumption that the scattering configuration is bulk-only --- the past lightsheets of $r_1$ and $r_2$ cannot hit the boundary after intersecting and before passing through $\mathcal{R}$, otherwise that point would be in the boundary scattering region $\hat{J}_{12 \rightarrow 12}.$\footnote{Recall, as explained in the introduction, that Theorem \ref{thm:main} is trivially true when the scattering configuration admits a boundary scattering point.} Finally, there may be cusps within $\mathcal{L}$ and $\mathcal{R}$ that arise from null generators colliding and leaving the lightsheet. These cusps extend down from the ridge along $\mathcal{L}$ and up from the contradiction surface $\mathcal{C}_{\Sigma}$ along $\mathcal{S}_{\Sigma}$. We label these cusps $\mathcal{B}_{\mathcal{L}}$ and $\mathcal{B}_{\mathcal{S}_{\Sigma}}$ in Figure \ref{fig:cusps}.

\begin{figure}
    \centering
    \includegraphics[scale=1.3]{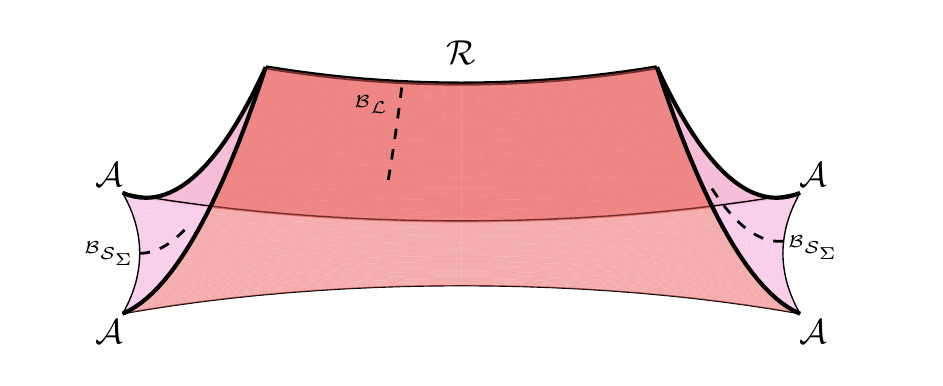} 
    \caption{The cusp structure of a generic null membrane. The two lightsheets that form the lift meet at a spacelike ridge $\mathcal{R}.$ The lift and slope meet at four spacelike seams $\mathcal{A}.$ The ridge and the seams are marked with heavy curves. The lift and slope may each have cusps formed from collisions of null generators, marked $\mathcal{B}_{\mathcal{L}}$ and $\mathcal{B}_{\mathcal{S}_{\Sigma}}$, respectively. These cusps extend down from the ridge along the lift, and up from the contradiction surface along the slope. They are marked here with dashed curves.}
    \label{fig:cusps}
\end{figure}

We may now use Stokes' theorem on the null membrane to demonstrate that the contradiction surface $\mathcal{C}_{\Sigma}$ has area less than or equal to the area of those surfaces. Heuristically, this follows from a focusing procedure in which one focuses the HRRT surfaces of $V_1$ and $V_2$ up the lift, along the ridge, and down the slope. More formally, the area difference may be computed as an integral of a codimension-$2$ volume form over the boundary of the null membrane:\footnote{While both areas on the left-hand side of this expression are infinite, their \emph{difference} can be treated consistently in a regulator-independent way using the methods of \cite{sorce2019cutoffs}.}
\begin{equation}
    \text{area}(\mathcal{C}_{\Sigma}) - \text{area}(\mathcal{E}_{V_1}^{\text{min}}\cup\mathcal{E}_{V_2}^{\text{min}})
        = \int_{\mathcal{C}_{\Sigma}} \bm{\tilde{\epsilon}} - \int_{\mathcal{E}_{V_1}^{\text{min}}\cup\mathcal{E}_{V_2}^{\text{min}}} \bm{\tilde{\epsilon}}.
\end{equation}
On the other hand, we may choose an affine parameter $\lambda$ for the slope such that the cusps $\mathcal{B}_{\mathcal{S}_{\Sigma}}$ and the contradiction surface $\mathcal{C}_{\Sigma}$ together form a surface of constant $\lambda$, and the seams $\mathcal{A}$ between the slope and the lift form another surface of constant $\lambda$. Then Stokes' theorem as implemented in equation \eqref{eq:stokes-area} yields the expression\footnote{The factor of two multiplying the area from the cusps comes from the fact that each point in a cusp is formed by the collision of two null generators.}
\begin{equation}
    \text{area}(\mathcal{C}_{\Sigma}) + 2\, \text{area}(\mathcal{B}_{\mathcal{S}_{\Sigma}})
        - \text{area}(\mathcal{A})
        = \int_{\mathcal{S}_{\Sigma}} \theta \bm{\epsilon},
\end{equation}
where $\theta$ is the past-directed expansion of $\mathcal{S}_{\Sigma}$ with respect to $\lambda$. Similarly, applying Stokes' theorem to $\mathcal{L}$ we obtain
\begin{equation}
    \text{area}(\mathcal{A}) + 2\, \text{area}(\mathcal{B}_{\mathcal{L}}) + 2\, \text{area}(\mathcal{R})
        - \text{area}(\mathcal{E}_{V_1}^{\text{min}}\cup\mathcal{E}_{V_2}^{\text{min}})
        = \int_{\mathcal{L}} \theta \bm{\epsilon}.
\end{equation}
Adding these two expressions yields the inequality
\begin{equation}
    \text{area}(\mathcal{C}_{\Sigma}) + 2 \left[ \text{area}(\mathcal{B}_{\mathcal{L}})
        + \text{area}(\mathcal{B}_{\mathcal{S}_{\Sigma}}) + \text{area}(\mathcal{R}) \right]
        - \text{area}(\mathcal{E}_{V_1}^{\text{min}}\cup\mathcal{E}_{V_2}^{\text{min}})
        = \int_{\mathcal{N}} \theta \bm{\epsilon} \leq 0,
\end{equation}
where we have used the fact that $\mathcal{L}$ is part of the future lightsheet of an extremal surface and $\mathcal{S}_{\Sigma}$ is part of the causal horizon of a point at infinity, meaning both have non-positive expansion.

We are led, finally, to the following inequality relating the disconnected-wedge HRRT surface $\mathcal{E}_{V_1}^{\text{min}}\cup\mathcal{E}_{V_2}^{\text{min}}$ and the contradiction surface $\mathcal{C}_{\Sigma}$:
\begin{equation} \label{eq:area-inequality}
    \text{area}(\mathcal{E}_{V_1}^{\text{min}}\cup\mathcal{E}_{V_2}^{\text{min}}) - \text{area}(\mathcal{C}_{\Sigma}) \geq  2 \left[ \text{area}(\mathcal{B}_{\mathcal{L}})
        + \text{area}(\mathcal{B}_{\mathcal{S}_{\Sigma}}) + \text{area}(\mathcal{R}) \right].
\end{equation}
To summarize, any complete achronal slice containing the disconnected-wedge HRRT surface also contains a contradiction surface $\mathcal{C}_{\Sigma}$ whose area is smaller by at least the sum of (i) the area of all null generators lost to cusps on $\mathcal{L}$, (ii) the area of all null generators lost to cusps on $\mathcal{S}_{\Sigma}$, and (iii) the area of all null generators lost to the ridge. This contradicts the assumption that the disconnected-wedge HRRT surface could have been found by a maximin procedure, thus proving Theorem \ref{thm:main}. 

This proof has the additional property that it gives a lower bound on the mutual information between $V_1$ and $V_2$ in terms of the geometry of the ridge. While the cusp terms $\text{area}(\mathcal{B}_{\mathcal{L}})$ and $\text{area}(\mathcal{B}_{\mathcal{S}_{\Sigma}})$ in inequality \eqref{eq:area-inequality} are highly dependent on the expansion of the null membrane, and $\text{area}(\mathcal{B}_{\mathcal{S}_{\Sigma}})$ depends on the choice of achronal slice, the ridge term $\text{area}(\mathcal{R})$ depends only on the geometry of a codimension-$2$ surface. Inequality \eqref{eq:area-inequality}, together with the HRRT formula \eqref{eq:HRRT}, implies the following lower bound on $I(V_1:V_2)$:
\begin{equation} \label{eq:bottom-seam}
    I(V_1:V_2) \geq \frac{\text{area}(\mathcal{R})}{2 G_N} + O(1).
\end{equation}
The size of the ridge $\mathcal{R}$ is controlled by the size of the scattering region $J_{12 \rightarrow 12}$; heuristically, equation \eqref{eq:bottom-seam} means that a larger scattering region implies larger mutual information between $V_1$ and $V_2$. 

Note that the proof techniques of this section are agnostic to time reversal; they imply also that the entanglement wedge of $\hat{J}_{12 \rightarrow 1} \cup \hat{J}_{12 \rightarrow 2}$ is connected.

Finally, we note that the inequality \eqref{eq:bottom-seam} is generically strict. This means that we shouldn't expect the converse of Theorem \ref{thm:main} to hold; the entanglement wedge can be connected even when no bulk scattering region exists. This is indeed the case --- an explicit counterexample is shown in Figure \ref{fig:counterexample}.

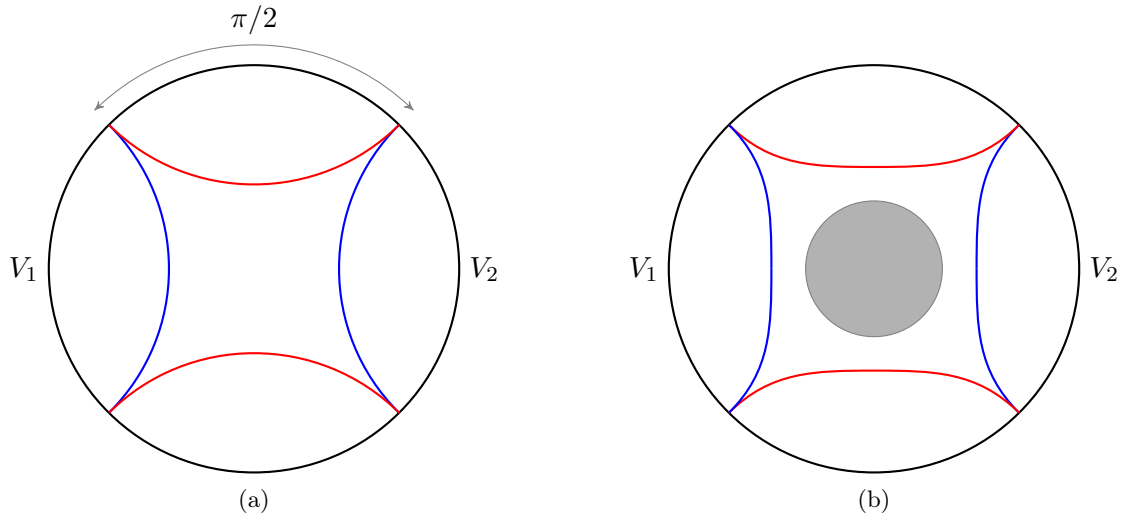
\begin{figure}
    \centering
    \subfloat[\label{fig:withoutmatter}]{
    \begin{tikzpicture}[scale=0.9]
    
    \draw[thick] (0,0) circle (3);
    
    \draw[blue, thick] (2.12, 2.12) to [out=-135,in=135] (2.12, -2.12);
    \draw[blue, thick] (-2.12, 2.12) to [out=-45,in=45] (-2.12, -2.12);
    
    \draw[red, thick] (-2.12, 2.12) to [out=-45,in=-135] (2.12, 2.12);
    \draw[red, thick] (-2.12, -2.12) to [out=45,in=135] (2.12, -2.12);
    
    \node[right] at (3,0) {$V_2$};
    \node[left] at (-3,0) {$V_1$};
    
    \draw[<->,gray,domain=135:45] plot ({3.3*cos(\x)},{3.3*sin(\x)});
    \node[above] at (0,3.3) {$\pi/2$};
    
    \end{tikzpicture}
    }
    \hfill
    \subfloat[\label{fig:withmatter}]{
    \begin{tikzpicture}[scale=0.9]
    
    \node[right] at (3,0) {$V_2$};
    \node[left] at (-3,0) {$V_1$};
    
    \draw[thick] (0,0) circle (3);
    \draw[fill=black!60!,opacity=0.5] (0,0) circle (1);
    
    \draw[blue, thick] (2.12, 2.12) to [out=-135,in=90] (1.5, 0);
    \draw[blue, thick] (1.5, 0) to [out=-90,in=135] (2.12, -2.12);
    
    \draw[red, thick] (2.12, 2.12) to [out=-135,in=0] (0, 1.5);
    \draw[red, thick] (0, 1.5) to [out=-180,in=-45] (-2.12, 2.12);
    
    \draw[red, thick] (2.12, -2.12) to [out=135,in=0] (0, -1.5);
    \draw[red, thick] (0, -1.5) to [out=-180,in=45] (-2.12, -2.12);
    
    \draw[blue, thick] (-2.12, 2.12) to [out=-45,in=90] (-1.5, 0);
    \draw[blue, thick] (-1.5, 0) to [out=-90,in=45] (-2.12, -2.12);
    
    \end{tikzpicture}
    }
    \caption{A counterexample to the converse of Theorem \ref{thm:main}. (a) A scattering configuration in vacuum AdS$_{2+1}$ is chosen so that the restriction of $V_1 \cup V_2$ to a spacelike slice is $\mathbb{Z}_2$ symmetric with its complement. The scattering region $J_{12 \rightarrow 12}$ consists of a single point, and the entanglement wedge of $V_1 \cup V_2$ becomes connected under any small enlargement of the region. (b) When spherically-symmetric matter with positive null energy is added to the bulk, time delay destroys the region $J_{12 \rightarrow 12}$. However, $V_1$ and $V_2$ are still symmetric with their complement; under an infinitesimally small perturbation, the entanglement wedge can become connected while the scattering region stays empty.}
    \label{fig:counterexample}
\end{figure}

\section{The scattering region is inside the entanglement wedge}
\label{sec:wedge}

As mentioned in the introduction, a consequence of Theorem \ref{thm:main} is that the scattering region $J_{12 \rightarrow 12}$ lies inside the connected entanglement wedge of $\hat{J}_{1\rightarrow 12} \cup \hat{J}_{2 \rightarrow 12}$. In this section, we present two short proofs of this fact based on the observation that entanglement wedge reconstruction places restrictions on bulk causal structure. The first proof directly uses the fact that the entanglement wedge lies outside the causal wedge in the bulk; the second proof frames the same observation in terms of boundary causality. Despite their different framing, the proofs are equivalent.

Note that the proofs in this section first assume that the entanglement wedge is connected, then use that assumption to show that the scattering region must be inside the entanglement wedge. We suspect it may be possible to prove Theorem \ref{thm:main} directly from entanglement wedge reconstruction by proving that the scattering region must \emph{always} be inside the entanglement wedge, which would in turn imply that the entanglement wedge is connected. We have as yet been unable to furnish a proof of Theorem \ref{thm:main} from this perspective, and it remains an interesting avenue for future work; we comment further on this point in the Discussion (Section \ref{sec:discussion}).

\subsection{Proof from the entanglement wedge containing the causal wedge}

Given two spacelike-separated domains of dependence $V_1$ and $V_2$ as in the statement of Theorem \ref{thm:main}, we may define the \emph{complementary} domain of dependence whose interior consists of all points spacelike-separated from $V_1$ and $V_2$. Just as in equation \eqref{eq:complementary-DOD}, this region is defined by
\begin{equation} \label{eq:complementary-DOD-2}
    X = \bar{[\hat{J}^+(V_1 \cup V_2) \cup \hat{J}^-(V_1 \cup V_2)]^c},
\end{equation}
where the overline denotes set closure. In $(2+1)$ bulk dimensions, where $V_1$ and $V_2$ have spacelike intervals as their Cauchy slices, $X$ has two connected components, which we call $X_1$ and $X_2$ as in Section \ref{sec:taskssection}. This is sketched in Figure \ref{fig:BB84taskWithRegions}.

In three bulk dimensions, the connected-wedge HRRT surface of $V_1 \cup V_2$ must consist of two disconnected geodesics. By the homology constraint, one of these must be the HRRT surface of $D(V_1 \cup V_2 \cup X_1)$, while the other is the HRRT surface of $D(V_1 \cup V_2 \cup X_2)$; by $D(S)$ we mean the domain of dependence of the set $S$.

Since the entanglement wedge always contains the causal wedge \cite{maximin}, and since $D(V_1 \cup V_2 \cup X_i)$ has as its future boundary the point $r_i$, the HRRT surface of $D(V_1 \cup V_2 \cup X_i)$ must be outside the past of $r_i$. We conclude that the HRRT surface of $V_1 \cup V_2$ must be outside the region $J^{-}(r_1) \cap J^{-}(r_2)$. It follows that the entanglement wedge of $V_1 \cup V_2$ must contain all points in $J^{-}(r_1) \cap J^{-}(r_2)$ that are in the future of $V_1$ and $V_2$, and hence must contain the scattering region $J_{12 \rightarrow 12}$.

\subsection{Proof from entanglement wedge reconstruction}

Define $X$ as in equation \eqref{eq:complementary-DOD-2} as the complementary domain of dependence to $V_1 \cup V_2$ on the boundary. Assuming that the global state is pure, the HRRT surfaces of $X$ and $V_1 \cup V_2$ are the same. The union of any two homology regions $\mathcal{R}_{X}$ and $\mathcal{R}_{V_1 \cup V_2}$ thus forms a complete achronal slice for the bulk.

Denote by $E_{W}(V_1 \cup V_2)$ and $E_W(X)$ the bulk entanglement wedges of $V_1 \cup V_2$ and $X$, respectively; i.e., we have
\begin{equation}
    E_{W}(V_1 \cup V_2) = D(\mathcal{R}_{V_1 \cup V_2})
    \quad \text{and} \quad
    E_{W}(X) = D(\mathcal{R}_{X}),
\end{equation}
where $D(S)$ is the domain of dependence of a set $S$. Assuming entanglement wedge reconstruction, the scattering region $J_{12 \rightarrow 12}$ cannot be in the timelike past of $E_W(X)$. If it were, then an operator in $J_{12 \rightarrow 12}$ could signal the interior of $E_W(X)$ and hence $X$; since $J_{12 \rightarrow 12}$ can be signalled by $c_1$ and $c_2$, this contradicts the assumption that the interior of $X$ is spacelike-separated from $V_1$ and $V_2$ on the boundary.

Similarly, $J_{12 \rightarrow 12}$ cannot be in the timelike future of $E_W(X)$. This part of the proof uses the fact the bulk is $(2+1)$-dimensional; as in the previous subsection, this implies that $X$ splits up into two disconnected components $X_1$ and $X_2$. Since $E_W(V_1 \cup V_2)$ is assumed to be connected, the entanglement wedge of $X$ must also split up into two disconnected components $E_W(X_1)$ and $E_W(X_2)$. This is sketched in Figure \ref{fig:bulksplit}. If $J_{12 \rightarrow 12}$ is in the future of $E_W(X)$, then it must be in the future of at least one of the entanglement wedges $E_W(X_i)$. This would imply that $J_{12 \rightarrow 12}$ and hence both $r_1$ and $r_2$ could be signalled by an operator in $E_W(X_i)$, which by entanglement wedge reconstruction is equivalent to an operator on $X_i$. On the boundary, however, each $X_i$ is restricted to signalling only one of the output points $r_i$; this is a feature of the $(2+1)$-dimensional causal structure. Hence $J_{12 \rightarrow 12}$ cannot be in the timelike future of $E_W(X)$

By exhausting all possibilities, the bulk scattering region $J_{12 \rightarrow 12}$ must be spacelike- or null-separated from $E_W(X)$; i.e., it must be inside $E_W(V_1 \cup V_2)$.

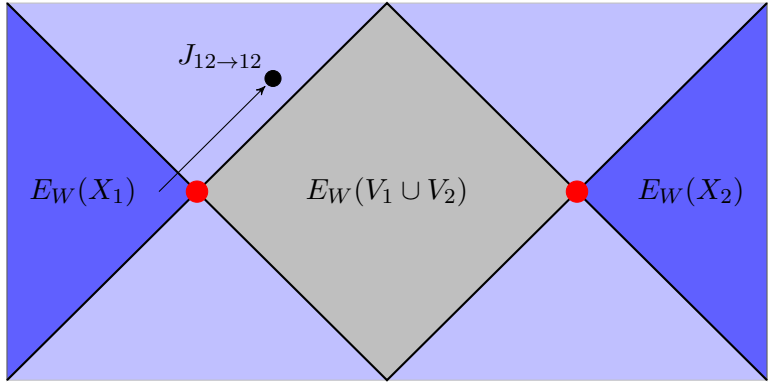
\begin{figure}
    \centering
    \begin{tikzpicture}[scale=1]
    
    \draw[fill=blue,opacity=0.25] (0,2.5) -- (-5,2.5) -- (-5,-2.5) -- (0,-2.5); 
    \draw[fill=blue,opacity=0.25] (0,2.5) -- (5,2.5) -- (5,-2.5) -- (0,-2.5);
    
    \draw[fill=blue,opacity=0.5] (-2.5,0) -- (-5,2.5) -- (-5,-2.5) -- (-2.5,0);
    \draw[fill=blue,opacity=0.5] (2.5,0) -- (5,2.5) -- (5,-2.5) -- (2.5,0);
    
    \draw[fill=lightgray,thick] (-2.5,0) -- (0,-2.5) -- (2.5,0) -- (0,2.5) -- (-2.5,0);
    
    \draw[thick] (-2.5,0) -- (-5,2.5);
    \draw[thick] (-2.5,0) -- (-5,-2.5);
    
    \draw[thick] (2.5,0) -- (5,2.5);
    \draw[thick] (2.5,0) -- (5,-2.5);
    
    \draw[red] plot [mark=*, mark size=4] coordinates{(-2.5,0)};
    \draw[red] plot [mark=*, mark size=4] coordinates{(2.5,0)};
    
    \node at (-4,0) {$E_W(X_1)$};
    \node at (4,0) {$E_W(X_2)$};
    \node at (0,0) {$E_W(V_1 \cup V_2)$};
    
    \draw[black] plot [mark=*,mark size=3] coordinates{(-1.5,1.5)};
    \draw[->] (-3,0) -- (-1.6,1.4);
    \node[above left] at (-1.5,1.5) {$J_{12\rightarrow 12}$};
    
    \end{tikzpicture}
    \caption{A cross-section of the bulk spacetime when the entanglement wedge of $V_1 \cup V_2$ is connected. The entanglement wedge of the complementary domain of dependence $X$ has two disconnected components $E_W(X_1)$ and $E_W(X_2)$. If $J_{12 \rightarrow 12}$ is in the future of one of these wedges $E_W(X_i)$, $X_i$ can signal $J_{12 \rightarrow 12}$ with a local operator in its entanglement wedge; this contradicts the fact that each $X_i$ can only signal one of the output points $r_i$.}
    \label{fig:bulksplit}
\end{figure}

\section{Discussion}\label{sec:discussion}

In this paper, we have extended the work of \cite{may2019quantum} concerning the \emph{connected wedge theorem} (Theorem \ref{thm:main}). We have refined the boundary arguments presented in \cite{may2019quantum} (Section \ref{sec:taskssection}), and have provided a proof of the connected wedge theorem using classical general relativity (Section \ref{sec:GR}). We have also provided a lower bound on the mutual information $I(V_1 : V_2)$ in terms of the geometry of the bulk scattering region $J_{12 \rightarrow 12}$ (eq.\eqref{eq:bottom-seam}), and proved that the bulk scattering region must always lie inside the entanglement wedge of $V_1 \cup V_2$ (Section \ref{sec:wedge}). Here we discuss various potential directions for future work.

\subsection{General relativity as a quantum information theorist}
\label{sec:gr-info-theorist}

In Section \ref{sec:taskssection} we studied a particular quantum computation occurring locally in the bulk, and reasoned about the resources necessary in the boundary to reproduce the same computation there nonlocally. We chose a simple task, called the $\textbf{B}_{84}^{\times n}$ task, and argued that for the computation to occur in the bulk with probability of success $p_{\text{suc}}=1-\epsilon$, the resource state $\rho_{V_1V_2}$ must satisfy
\begin{align}
    \frac{1}{2}I(V_1:V_2)_\rho \geq - \log [2(\epsilon+\beta^n)].
\end{align}
Further, we pointed out that this bound is essentially tight --- via a simple teleportation protocol we can complete the task with $p_{\text{suc}}=1$ and $I(V_1:V_2)/2=n$. 

Suppose however that we asked for a different computation to occur in the bulk picture. For instance, suppose we specify some simple, non-Clifford unitary $U$. In that case, no simple procedure for completing that computation nonlocally is known. The only known procedures consume exponential entanglement \cite{chakraborty2015attack,beigi2011simplified,dolev2019constraining} --- for example, using port-teleportation \cite{ishizaka2008asymptotic}. More precisely, completing the task with $p_{\text{suc}}=1-\epsilon$ on $n$ qubits requires a mutual information of order $n2^n/\epsilon$. However, the strongest lower bounds on the required mutual information are linear in $n$ \cite{may2019quantum}. A longstanding question in quantum cryptography is whether the exponentially costly procedures can be improved to meet the linear lower bound.

General relativity answers this quantum information theoretic question. The AdS/CFT dictionary gives a nonlocal implementation of these more general computations using only linear entanglement, which is an exponential improvement over all known constructions. To see this, we need only consider building a quantum computer that implements $U$ in the bulk scattering region $J_{12\rightarrow 12}$. Arranging for the scattering process to complete the computation in the bulk implies the same computation occurs in the boundary. As in the discussion around the tasks argument for Theorem \ref{thm:main}, the number of qubits $n$ involved in the computation can be taken to be $O(1/G_N^{\alpha})$ for $\alpha\in [0,1)$. Meanwhile, the mutual information between the input regions is $O(1/G_N)$, and so linear in the number of qubits. In future work we hope to construct explicit circuit-based implementations of nonlocal computations based on the AdS/CFT dictionary. 

\subsection{Higher-dimensional connected wedge theorems}
\label{sec:higher-d}

We have largely discussed Theorem \ref{thm:main} in the context of asymptotically AdS$_{2+1}$ spacetimes. The situation in higher dimensional theories of gravity is more subtle. A key feature of the connected wedge theorem in AdS$_{2+1}$ is that a given set of boundary points $\{c_1, c_2, r_1, r_2\}$ may have a nontrivial bulk scattering region $J_{12\rightarrow12}$ even when the boundary scattering region $\hat{J}_{12\rightarrow12}$ is empty. In higher dimensions this feature is generally absent: in fact, in any asymptotically global AdS spacetime satisfying the null energy condition, the non-emptiness of $J_{12 \rightarrow 12}$ implies the non-emptiness of $\hat{J}_{12\rightarrow 12}.$ This was observed earlier in the context of scattering in higher dimensions \cite{GGP, HPPS, Penedones, MSZ}. This observation makes Theorem \ref{thm:main} trivial in dimensions $d > 2$: when both bulk and boundary scattering regions are nontrivial, the input regions $V_1$ and $V_2$ overlap, and so the connectedness of their joint entanglement wedge is unsurprising. 

However, there exist asymptotically \emph{locally} AdS spacetimes\footnote{A spacetime is said to be ``asymptotically global AdS$_{d+1}$'' if it has a conformal boundary such that (i) the boundary has the same structure as that of global AdS, i.e., it has topology $S_{d-1} \times \mathbb{R}$ with the standard conformal metric, and (ii) the AdS equation $R_{ab} = -d g_{ab}$ is satisfied at leading order near the boundary. It is said to be ``asymptotically locally AdS$_{d+1}$ if condition (ii) but not (i) is satisfied.} in which the existence of a bulk scattering region does not imply the existence of a boundary scattering region. One such spacetime is the \emph{AdS soliton} \cite{horowitz1998ads}, whose boundary is topologically distinct from that of AdS. In $(3+1)$ dimensions, there exist $2$-to-$2$ bulk-only scattering configurations on the AdS soliton. Theorem \ref{thm:main} thus applies nontrivially.

We may also ask if there is some modifcation of the connected wedge theorem which applies nontrivially to asymptotically global AdS$_{d+1}$. Following the higher dimensional scattering story, it is natural to introduce additional output points $r_1,...,r_d$. Doing so it is possible to arrange the $c_i$ and $r_i$ such that $J_{12\rightarrow 1...d}$ is nonempty, but $\hat{J}_{12\rightarrow 1...d}$ is empty. In this case we return to a setting where the bulk geometry has causal features not present in the boundary, and it is plausible that this imposes requirements on boundary correlations. We have not been able to resolve this question, but offer some comments and partial progress below. 

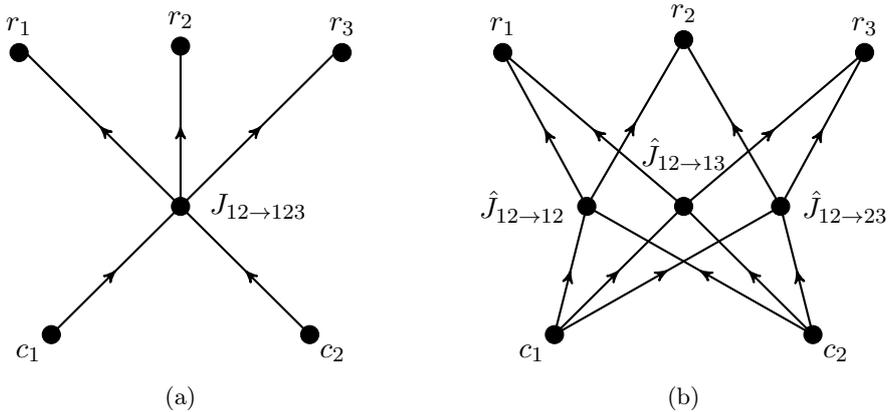
\begin{figure}
\centering
    \subfloat[\label{fig:2-3-bulk}]{
    \begin{tikzpicture}[scale=1.7]
    
    \draw[black] plot [mark=*, mark size=2] coordinates{(0,0)};
    \node[right=0.25cm] at (0,0) {${J}_{12\rightarrow 123}$};
    
    \draw[black] plot [mark=*, mark size=2] coordinates{(-1,-1)};
    \node[below left] at (-1,-1) {$c_1$};
    \draw[black] plot [mark=*, mark size=2] coordinates{(1,-1)};
    \node[below right] at (1,-1) {$c_2$};
    
    \draw[thick,postaction={on each segment={mid arrow}}] (-1,-1) -- (0,0);
    \draw[thick,postaction={on each segment={mid arrow}}] (1,-1) -- (0,0);
    
    \draw[black] plot [mark=*, mark size=2] coordinates{(-1.25,1.2)};
    \draw[black] plot [mark=*, mark size=2] coordinates{(0,1.25)};
    \draw[black] plot [mark=*, mark size=2] coordinates{(1.25,1.2)};
    
    \node[above=0.1cm] at (-1.25,1.2) {$r_1$};
    \node[above=0.1cm] at (0,1.25) {$r_2$};
    \node[above=0.1cm] at (1.25,1.2) {$r_3$};
    
    \draw[thick,postaction={on each segment={mid arrow}}] (0,0) -- (-1.25,1.25);
    \draw[thick,postaction={on each segment={mid arrow}}] (0,0) -- (0,1.25);
    \draw[thick,postaction={on each segment={mid arrow}}] (0,0) -- (1.25,1.25);
    
    \end{tikzpicture} 
    }
    \hspace{1cm}
    \subfloat[\label{fig:2-3-boundary}]{
    \begin{tikzpicture}[scale=1.7]
    
    \coordinate (c1) at (-1,-1);
    \coordinate (c2) at (1,-1);
    
    \coordinate (r1) at (-1.4,1.2);
    \coordinate (r2) at (0,1.3);
    \coordinate (r3) at (1.4,1.2);
    
    \coordinate (p13) at (0,0);
    \coordinate (p12) at (-0.75,0);
    \coordinate (p23) at (0.75,0);
    
    \draw[black] plot [mark=*, mark size=2] coordinates{(r1)};
    \draw[black] plot [mark=*, mark size=2] coordinates{(r2)};
    \draw[black] plot [mark=*, mark size=2] coordinates{(r3)};
    
    \draw[black] plot [mark=*, mark size=2] coordinates{(c1)};
    \node[below left] at (c1) {$c_1$};
    \draw[black] plot [mark=*, mark size=2] coordinates{(c2)};
    \node[below right] at (c2) {$c_2$};
    
    \draw[black] plot [mark=*, mark size=2] coordinates{(p12)};
    \draw[black] plot [mark=*, mark size=2] coordinates{(p13)};
    \draw[black] plot [mark=*, mark size=2] coordinates{(p23)};
    
     \draw[thick,postaction={on each segment={mid arrow}}] (c1) -- (p13);
     \draw[thick,postaction={on each segment={mid arrow}}] (c1) -- (p12);
     \draw[thick,postaction={on each segment={mid arrow}}] (c1) -- (p23);
     
     \draw[thick,postaction={on each segment={mid arrow}}] (c2) -- (p13);
     \draw[thick,postaction={on each segment={mid arrow}}] (c2) -- (p12);
     \draw[thick,postaction={on each segment={mid arrow}}] (c2) -- (p23);
     
     \draw[thick,postaction={on each segment={mid arrow}}] (p12) -- (r1);
     \draw[thick,postaction={on each segment={mid arrow}}] (p12) -- (r2);
     
     \draw[thick,postaction={on each segment={mid arrow}}] (p13) -- (r1);
     \draw[thick,postaction={on each segment={mid arrow}}] (p13) -- (r3);
     
     \draw[thick,postaction={on each segment={mid arrow}}] (p23) -- (r2);
     \draw[thick,postaction={on each segment={mid arrow}}] (p23) -- (r3);
     
    \node[above=0.1cm] at (r1) {$r_1$};
    \node[above=0.1cm] at (r2) {$r_2$};
    \node[above=0.1cm] at (r3) {$r_3$};
    
    \node[above=0.35cm] at (p13) {$\hat{J}_{12\rightarrow 13}$};
    \node[left=0.15cm] at (p12) {$\hat{J}_{12 \rightarrow 12}$};
    \node[right=0.15cm] at (p23) {$\hat{J}_{12 \rightarrow 23}$};
      
    \end{tikzpicture}
    }
    \caption{The causal structure of a $2$-to-$3$ bulk scattering configuration with no boundary scattering region. (a) The causal structure in the bulk, where each input point may signal an intermediary region that may signal all three output points. (b) The causal structure in the boundary, where each input point may signal intermediary regions that may signal \emph{two} of the three output points.}
    \label{fig:2-3-scattering}
\end{figure}

Consider the case of global AdS$_{3+1}$. Suppose one has chosen a configuration of points $c_1, c_2, r_1, r_2, r_3$ on the boundary of vacuum AdS$_{3+1}$ where the boundary region $\hat{J}_{12 \rightarrow 123}$ is empty, but the bulk region ${J}_{12 \rightarrow 123}$ is nonempty. The bulk causal structure of this setup is shown heuristically in Figure \ref{fig:2-3-bulk}: both input points $c_1$ and $c_2$ can signal a bulk region $J_{12\rightarrow123}$, which can in turn signal each of the output points $r_1, r_2$, and $r_3$. The boundary structure has $\hat{J}_{12\rightarrow 123}$ empty, but $\hat{J}_{12\rightarrow j k}$ nonempty for any two output points $r_{j,k}.$

There are now two distinct classes of boundary causal regions associated with the scattering configuration: the two ``input regions'' $\hat{J}_{1 \rightarrow 123}$ and $\hat{J}_{2 \rightarrow 123}$, which are analogous to the input regions in the $2$-to-$2$ scattering configuration, and also the three ``intermediate regions'' $\hat{J}_{12 \rightarrow 12},$ $\hat{J}_{12 \rightarrow 13}$, and $\hat{J}_{12 \rightarrow 23}.$ We can imagine setting up generalizations of the quantum tasks discussed in Section \ref{sec:taskssection} on these configurations of input and output points. Again in the bulk the required quantum computations may occur locally, in the region $J_{12\rightarrow 123}$, while in the boundary they must occur nonlocally. It is natural to suspect that some combination of the input and intermediate regions must share large correlations in order to perform these computations nonlocally. 

One way to resolve this question is at the level of quantum information theory --- we can attempt to prove that entanglement is necessary, or attempt to construct procedures that perform the computation nonlocally without entanglement. We go some way towards this in Appendix \ref{sec:onetimepad} by showing that in cases where we can hope to build simple procedures, entanglement is unnecessary. A second approach is to address this question at the level of general relativity --- by studying the extremal surfaces attached to the input and intermediate regions, we can understand when these regions share large correlations. Because the input and intermediate regions are not symmetric about any time slice, this is a nontrivial task that we leave to future work. It is intriguing that solving a problem in general relativity --- finding disconnected entanglement wedges --- would have a strong implication in quantum information theory. In particular finding disconnected wedges would imply all quantum tasks on the $2$-to-$3$ geometry can be performed without entanglement.

\subsection{Alternative proofs of the connected wedge theorem}
\label{sec:alternative-proofs}

As mentioned in the introduction and in Section \ref{sec:wedge}, we suspect it may be possible to prove Theorem \ref{thm:main} directly from entanglement wedge reconstruction. If one could argue from holographic considerations that the scattering region $J_{12 \rightarrow 12}$ lies within the entanglement wedge $E_W(V_1 \cup V_2)$, this would imply that the entanglement wedge is connected. This remains an open question.

Another approach toward Theorem \ref{thm:main} would be to make direct contact with the correlation function singularities referenced in the introduction. While the existence of bulk-only scattering configurations that motivate the connected wedge theorem was first identified in the study of boundary $n$-point functions, we have not used any of the technical results from that work in our proofs. It would be interesting to find a way to show that $V_1$ and $V_2$ must have $O(1/G_N)$ mutual information directly from the existence of a perturbative singularity in four-point functions on $\{c_1, c_2, r_1, r_2\}$. 

\subsection{Metric reconstruction}
\label{sec:metric-reconstruction}

It is of considerable interest to determine how the bulk metric of general relativity is encoded in the boundary under the AdS/CFT dictionary. An early proposal for reconstructing the bulk metric from boundary data involved \emph{light-cone cuts}: the restriction to the spacetime boundary of the light cone of a bulk point. It was shown in \cite{LCC1, LCC2} that knowledge of light-cone cuts, which can be determined from singularities in boundary correlators, can be used to reconstruct the metric. Another approach to metric reconstruction is to use the geometry of bulk extremal surfaces, whose areas are known on the boundary by the HRRT formula \eqref{eq:HRRT}. It has been argued that in many cases the extremal surface areas determine the bulk metric \cite{rigidity1}.

The connected wedge theorem gives a nontrivial geometric relation between the causal structure of a spacetime and the areas of its extremal surfaces. It would be interesting to understand how, if at all, the intuition gleaned from Theorem \ref{thm:main} might assist in relating the two extant approaches to metric reconstruction.

\acknowledgments{We thank Patrick Hayden, Mark Van Raamsdonk, Jamie Sully, Veronika Hubeny, Adam Levine, Mukund Rangamani, David Wakeham, Chris Waddell, and Aron Wall for useful conversations. We especially thank Mark Wilde for pointing out an error in version 1 of this paper, and Patrick Hayden for very useful conversations during the writing of version 2. AM thanks the Stanford Institute of Theoretical Physics for hospitality during an extended visit, during which this work was completed. Portions of this work were completed at the Yukawa Institute for Theoretical Physics during workshop YITP-T-19-03. The authors are supported in part by the Simons Foundation It from Qubit collaboration. AM is further supported by C-GSM and MSFSS awards given by the National Science and Engineering Research Council. GP and JS are further supported by AFOSR (FA9550-16-1-0082) and DOE Award No.\ DE-SC0019380.}

\appendix

\section{Summary of some information measures and their properties}\label{appendix:relativeentropyandTrace}

In this appendix we recall various standard inequalities and distinguishability measures, and prove a few lemmas we need in the main article.

We begin by recalling the definition of the trace distance,
\begin{definition}\textbf{Trace distance:} $||\rho-\sigma||_1 := \tr|\rho-\sigma|$
\end{definition}
The trace distance is zero if and only if $\rho=\sigma$, and has a maximal value of $2$. 

The trace distance has an operational interpretation. Suppose we are handed a quantum system whose reduced density matrix is $\rho$ with probability $1/2$, or $\sigma$ with probability $1/2$. Then our probability of successfully distinguishing $\rho$ from $\sigma$, optimized over all possible measurements, is
\begin{align}\label{eq:maxdistinguish}
    p_{\text{dist}}^{\text{max}}(\rho,\sigma) = \frac{1}{2} + \frac{1}{4}||\rho-\sigma||_1
\end{align}
We can use this operational interpretation to prove a continuity bound on success probabilities of any task.
\begin{lemma}\label{lemma:psuccontinuity}
Consider a task which takes as input a quantum system $A$. Then the success probability of the task, call it $p_{\text{suc}}$, satisfies the continuity bound
\begin{align}\label{eq:psuccontinuity}
    |p_{\text{suc}}(\rho_A) - p_{\text{suc}}(\sigma_A)| \leq \frac{1}{2}|| \rho_A-\sigma_A ||_1.
\end{align}
\end{lemma}
\begin{proof}
The bound follows by considering the task as a means of distinguishing $\rho$ and $\sigma$. The task is given the input system $\rho$ with probability $1/2$ or $\sigma$ with probability $1/2$. If the task succeeds, we declare the state to be $\rho$. If the task fails, we declare the state to be $\sigma$. Then
\begin{align}
    p_{\text{dist}}(\rho,\sigma) &= \frac{1}{2} p(\text{task succeeds}\,|\,\text{state is} \,\rho) + \frac{1}{2} p(\text{task fails}\,|\,\text{state is} \,\sigma) \nonumber \\
    &= \frac{1}{2} p_{\text{suc}}(\rho) + \frac{1}{2}(1 - p_{\text{suc}}(\sigma))
\end{align}
Now we use that the probability of success of this particular method of distinguishing between $\rho$ and $\sigma$ is bounded above by the maximal probability, as expressed in terms of the trace distance in equation \eqref{eq:maxdistinguish}. This leads to
\begin{align}
    p_{\text{suc}}(\rho) - p_{\text{suc}}(\sigma) \leq \frac{1}{2}||\rho-\sigma||_1
\end{align}
reversing the roles of $\rho$ and $\sigma$ (declaring the state to be $\rho$ when the task fails, and to be $\sigma$ when it succeeds) we obtain the second inequality
\begin{align}
    p_{\text{suc}}(\sigma) - p_{\text{suc}}(\rho) \leq \frac{1}{2}||\rho-\sigma||_1
\end{align}
which establishes \eqref{eq:psuccontinuity}. 
\end{proof}

Next we recall the fidelity and its relation to the trace distance.
\begin{definition}\textbf{Fidelity:} $ F(\rho,\sigma) = \left[\tr \sqrt{\sqrt{\rho}\sigma\sqrt{\rho}} \right]^2$
\end{definition}
The trace distance and fidelity are related in the following lemma \cite{fuchs1999cryptographic,wilde2013quantum}
\begin{lemma} Two states $\rho$, $\sigma$, that are close in trace distance have large fidelity. Similarly, two states that have large fidelity are close in trace distance. In particular, 
\begin{align}
    1-\sqrt{F(\rho,\sigma)}\leq \frac{1}{2}|| \rho-\sigma||_1 \leq \sqrt{1-F(\rho,\sigma)}.
\end{align}
\end{lemma}

\section{Entanglement-hiding nonlocal computations for \texorpdfstring{$2$}{TEXT}-to-\texorpdfstring{$2$}{TEXT} scattering}\label{sec:22appendix}

In section \ref{sec:taskssection}, we commented that if Alice can make use of the complementary regions $X_1$ and $X_2$ in her protocol then it becomes possible to complete the $\textbf{B}_{84}$ task nonlocally without large correlations between the input regions. In this appendix we briefly explain how this is possible. 

The protocol is based on the quantum one-time pad \cite{ambainis2000private}, which is a tool from quantum cryptography used to hide quantum information using classical information. A one-time pad is a set of unitaries $\{U_k\}$ such that averaging over $k$ returns the maximally mixed state,
\begin{align}
    \frac{1}{|k|}\sum_k U_k\rho_A U^\dagger_k = \frac{\mathcal{I}}{d_A},
\end{align}
for any $\rho_A$. For instance, when $A$ is a single qubit, the Pauli set $\{P_k\}=\{\mathcal{I},X,Y,Z\}$ works as a one-time pad. We may then define the state
\begin{align}\label{eq:OTPstate}
    \ket{\psi} = \frac{1}{|k|}\sum_k \ket{k}_{x_1}\otimes\ket{k}_{x_2}\otimes (\mathcal{I}\otimes P_k) \ket{\Psi^+}_{v_1v_2}
\end{align}
on $\mathbb{C}^4 \otimes \mathbb{C}^4 \otimes (\mathbb{C}^2)^{\otimes 2}$, where $\ket{\Psi^+}$ is the maximally entangled state $(\ket{00}+\ket{11})/\sqrt{2}$. Treating the two single-qubit subsystems as a composite four-dimensional subsystem, $\ket{\psi}$ is a three-party GHZ state for systems $x_1$, $x_2$, and $v_1 v_2$.

The state $\ket{\psi}$ has zero mutual information between $v_1$ and $v_2$, but can still be used as a resource for the $\textbf{B}_{84}$ task as follows. First, we send systems $x_i$ to spacetime regions $X_i$, and systems $v_i$ to spacetime regions $V_i$ (cf. Figure \ref{fig:BB84taskWithRegions}). When the nonlocal circuit of Figure \ref{fig:nonlocalcircuit} is applied to the state $\mathcal{I}\otimes P_k\ket{\Psi^+}$, the classical measurement outcomes of the circuit are related to $q$ and $b$ in a way that depends on $P_k=X^mZ^n$, in particular $(-1)^b=(s_1)^q(s_2)^{1-q}s_3(-1)^{m(1-q)+nq}$. If we apply the nonlocal circuit to $\ket{\psi}$ in systems $v_1 v_2$, and measure $\ket{\psi}$ in the $k$-basis in the $x_i$ subsystems, the combined classical output of both measurements is sufficient to complete the $\textbf{B}_{84}$ task with zero information between $v_1$ and $v_2$.

Importantly, in the context of holography, the three-party GHZ entanglement between $x_1$, $x_2$ and $v_1 v_2$ cannot be created by acting locally at $c_1$ and $c_2$, since $X_1$ and $X_2$ are not in the future lightcone of either $c_1$ or $c_2$. Instead, the entanglement must already exist in the initial semiclassical state. As discussed in Section \ref{sec:22connectedwedge}, it is not expected that this is ever the case.

\section{Entanglement-hiding nonlocal computations for \texorpdfstring{$2$}{TEXT}-to-\texorpdfstring{$3$}{TEXT} scattering}
\label{sec:onetimepad}

\begin{figure}
\centering
    \begin{tikzpicture}[scale=1.6]
    
    \coordinate (c1) at (-1,-1);
    \coordinate (c2) at (1,-1);
    
    \coordinate (r1) at (-1.4,1.2);
    \coordinate (r2) at (0,1.3);
    \coordinate (r3) at (1.4,1.2);
    
    \coordinate (p13) at (0,0);
    \coordinate (p12) at (-0.75,0);
    \coordinate (p23) at (0.75,0);
    
    \draw[black] plot [mark=*, mark size=2] coordinates{(r1)};
    \draw[black] plot [mark=*, mark size=2] coordinates{(r2)};
    \draw[black] plot [mark=*, mark size=2] coordinates{(r3)};
    
    \draw[black] plot [mark=*, mark size=2] coordinates{(c1)};
    \node[below left] at (c1) {$c_1$};
    \draw[black] plot [mark=*, mark size=2] coordinates{(c2)};
    \node[below right] at (c2) {$c_2$};
    
    \draw[black] plot [mark=*, mark size=2] coordinates{(p12)};
    \draw[black] plot [mark=*, mark size=2] coordinates{(p13)};
    \draw[black] plot [mark=*, mark size=2] coordinates{(p23)};
    
     \draw[thick,postaction={on each segment={mid arrow}}] (c1) -- (p13);
     \draw[thick,postaction={on each segment={mid arrow}}] (c1) -- (p12);
     \draw[thick,postaction={on each segment={mid arrow}}] (c1) -- (p23);
     
     \draw[thick,postaction={on each segment={mid arrow}}] (c2) -- (p13);
     \draw[thick,postaction={on each segment={mid arrow}}] (c2) -- (p12);
     \draw[thick,postaction={on each segment={mid arrow}}] (c2) -- (p23);
     
     \draw[thick,postaction={on each segment={mid arrow}}] (p12) -- (r1);
     \draw[thick,postaction={on each segment={mid arrow}}] (p12) -- (r2);
     
     \draw[thick,postaction={on each segment={mid arrow}}] (p13) -- (r1);
     \draw[thick,postaction={on each segment={mid arrow}}] (p13) -- (r3);
     
     \draw[thick,postaction={on each segment={mid arrow}}] (p23) -- (r2);
     \draw[thick,postaction={on each segment={mid arrow}}] (p23) -- (r3);
     
    \node[above=0.1cm] at (r1) {$r_1$};
    \node[above=0.1cm] at (r2) {$r_2$};
    \node[above=0.1cm] at (r3) {$r_3$};
    
    \node[above=0.35cm] at (p13) {$\hat{J}_{12\rightarrow 13}$};
    \node[left=0.15cm] at (p12) {$\hat{J}_{12 \rightarrow 12}$};
    \node[right=0.15cm] at (p23) {$\hat{J}_{12 \rightarrow 23}$};
    
    \draw[blue,postaction={on each segment={mid arrow}}] (c1) to  [out=100,in=-90] (r1);
    
    \draw[blue,postaction={on each segment={mid arrow}}] (c2) to  [out=80,in=-90] (r3);
    
    \draw[blue,postaction={on each segment={mid arrow}}] (c1) to  [out=105,in=-130] (-1.8,1.7);
    \draw[blue,postaction={on each segment={mid arrow}}] (-1.8,1.7) to  [out=45,in=135] (r2);
    
    \draw[blue,postaction={on each segment={mid arrow}}] (c2) to  [out=75,in=-50] (1.8,1.7);
    \draw[blue,postaction={on each segment={mid arrow}}] (1.8,1.7) to  [out=140,in=45] (r2);
    
    \draw[blue,postaction={on each segment={mid arrow}}] (c1) to  [out=110,in=-135] (-2,1.9);
    \draw[blue,postaction={on each segment={mid arrow}}] (-2,1.9) to  [out=40,in=140] (r3);

    \draw[blue,postaction={on each segment={mid arrow}}] (c2) to  [out=70,in=-45] (2,1.9);
    \draw[blue,postaction={on each segment={mid arrow}}] (2,1.9) to  [out=140,in=40] (r1);    
    \end{tikzpicture}
    
    \caption{Schematic diagram of the boundary causal structure in $2$-to-$3$ holographic scattering. The regions $\hat{J}_{12\rightarrow jk}$ are all nonempty, meaning it is possible for signals to travel from $c_1$ and $c_2$ meet, and then get to either of $r_j$ or $r_k$. There are also causal curves that travel directly from $c_i$ to $r_j$ without passing through any of the regions $\hat{J}_{12\rightarrow jk}$. This fact is crucial for the entanglement-free procedures.}
    \label{fig:2-3-scattering-withblue}
\end{figure}
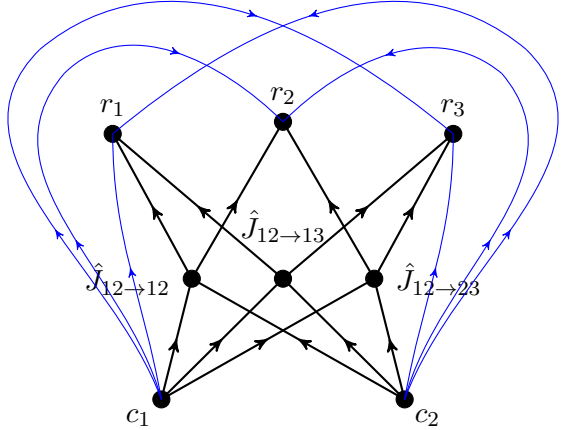

In Section \ref{sec:higher-d} we discussed $2$-to-$3$ scattering arrangements of points. Setting up quantum information theoretic tasks of the type discussed in Section \ref{sec:taskssection} on such arrangements of points seems a natural approach to extending the connected wedge theorem to asymptotically global AdS$_{d+1}$ spacetimes for $d>2$. We point out here however that, at least for those tasks we can most easily understand, entanglement is not necessary between any of the causal regions associated with the scattering. We take this as suggesting the connected wedge theorem does not extend in this obvious way to higher dimensions, though other modifications may be possible. 

For concreteness consider the following task. At $c_1$, a quantum system $R$ is input; at $c_2$ a classical trit $q\in \{1,2,3\}$ is input. The requirement is that $R$ be output at $r_q$. In the bulk picture, where $J_{12\rightarrow 123}$ is nonempty, there is an obvious procedure to complete this --- bring $R$ and $q$ into $J_{12\rightarrow 123}$ and route $R$ towards the appropriate output point based on $q$. In the boundary, where $\hat{J}_{12\rightarrow 123}$ is empty, there is a procedure which does not require entanglement between the input regions $\hat{J}_{i\rightarrow 123}$ nor between the intermediate regions $\hat{J}_{12\rightarrow jk}$. 

A naive procedure is as follows. At $c_1$, record $R$ into an error correcting code that has three shares and corrects one erasure error. Call the three shares $ABC$, and label the system purifying the state on $R$ by $\bar{R}$, so that
\begin{align}\label{eq:QECencoding}
    \ket{\psi}_{\bar{R}R}\rightarrow \ket{\Psi}_{\bar{R}ABC}.
\end{align}
Now send $A$ to $\hat{J}_{12\rightarrow 12},$ $B$ to $\hat{J}_{12\rightarrow 13}$, and $C$ to $\hat{J}_{12\rightarrow 23}$. At $c_2$, create copies of $q$ and send one copy to each of the $\hat{J}_{12\rightarrow jk}$. At each of the $\hat{J}_{12\rightarrow jk}$, forward the share of the error correcting code to $r_q$ if $q\in \{j,k\}$. This will be possible at two of the three regions, so that two of the three shares will arrive at the correct output point $r_q$. The quantum system $R$ can then be recovered from these two shares. 

The procedure above avoids using entanglement between the input regions $\hat{J}_{i\rightarrow 123}$. However, quantum error correcting codes record information into highly entangled states. In particular the above procedure would create entanglement among the boundary regions $\hat{J}_{12\rightarrow jk}$. This entanglement too can be avoided however. To do this we exploit an additional feature of the boundary causal structure shown in Figure \ref{fig:2-3-scattering-withblue}: In the boundary, there are causal curves from the input points to the output points that avoid the intermediate regions. 

These intermediate region avoiding curves can be used to hide the entanglement between the intermediate regions. Begin by again recording $R$ into an error correcting code as in equation \eqref{eq:QECencoding}, but now additionally exploit the quantum one-time pad \cite{ambainis2000private} (see appendix \ref{sec:22appendix}) to encrypt each of the shares $A$, $B$ and $C$. Thus at $c_2$ we encode $R$ according to
\begin{align}
    \ket{\psi}_{\bar{R}R} \rightarrow \sum_{k_a,k_b,k_c\in \{0,1\}^{\times 2n}}\ket{k_ak_b}_{X_1}\otimes \ket{k_ak_c}_{X_2}\otimes \ket{k_bk_c}_{X_3} \otimes [U_A^{k_a} U_B^{k_b} U_C^{k_c}] \ket{\Psi}_{RABC}
\end{align}
The one-time pad provides a way to choose the sets of unitaries $\{ U^k\}$ such that averaging over $k$ produces the maximally mixed state. The procedure now is to send the $A$, $B$ and $C$ shares through the intermediate regions as before, but now additionally send the systems $X_i$ to the corresponding $r_i$ directly along the intermediate region avoiding curves. Now the intermediate regions remain unentangled, since the state on $ABC$ is maximally mixed. Following the same routing procedure as before, two shares from the error correcting code will again arrive at the appropriate output points. The systems $X_i$ can be used to undo the action of the one-time pad before recovering $R$ from the shares of the error correcting code.

We have also studied generalizations of the $\textbf{B}_{84}$ task to the $2$-to-$3$ setting and found entanglement-free protocols in that case. As with the ``routing'' task described above, these protocols involve use of a one-time pad. Those protocols however rely on the only quantum operation performed being a Hadamard, which happens to be in the Clifford group. Clifford operations have simple conjugation properties with Pauli operations, and we can chose the one-time pad to involve only Paulis. These coincidences allow an entanglement-free procedure to be designed simply. For more general tasks we do not know how to construct entanglement-free protocols, or if they exist. It is interesting that an extremal surface calculation may resolve this question. 

\bibliographystyle{JHEP}
\bibliography{biblio}

\providecommand{\href}[2]{#2}\begingroup\raggedright\begin{thebibliography}{10}

\bibitem{GGP}
M.~Gary, S.~B. Giddings, and J.~Penedones, {\it Local bulk {S-matrix elements}
  and {CFT} singularities},  {\em Physical Review D} {\bf 80} (2009) 085005,
  [\href{http://arxiv.org/abs/0903.4437}{{\tt arXiv:0903.4437}}].

\bibitem{HPPS}
I.~Heemskerk, J.~Penedones, J.~Polchinski, and J.~Sully, {\it Holography from
  conformal field theory},  {\em Journal of High Energy Physics} {\bf 2009}
  (2009), no.~10 079, [\href{http://arxiv.org/abs/0907.0151}{{\tt
  arXiv:0907.0151}}].

\bibitem{Penedones}
J.~Penedones, {\it Writing {CFT} correlation functions as {AdS} scattering
  amplitudes},  {\em Journal of High Energy Physics} {\bf 2011} (2011), no.~3
  [\href{http://arxiv.org/abs/1011.1485}{{\tt arXiv:1011.1485}}].

\bibitem{MSZ}
J.~Maldacena, D.~Simmons-Duffin, and A.~Zhiboedov, {\it Looking for a bulk
  point},  {\em Journal of High Energy Physics} {\bf 2017} (2017), no.~13
  [\href{http://arxiv.org/abs/1509.03612}{{\tt arXiv:1509.03612}}].

\bibitem{RT}
S.~Ryu and T.~Takayanagi, {\it Aspects of holographic entanglement entropy},
  {\em Journal of High Energy Physics} {\bf 2006} (2006) 045,
  [\href{http://arxiv.org/abs/hep-th/0605073}{{\tt hep-th/0605073}}].

\bibitem{HRT}
V.~E. Hubeny, M.~Rangamani, and T.~Takayanagi, {\it A covariant holographic
  entanglement entropy proposal},  {\em Journal of High Energy Physics} {\bf
  2007} (2007) 062, [\href{http://arxiv.org/abs/0705.0016}{{\tt
  arXiv:0705.0016}}].

\bibitem{FLM}
T.~Faulkner, A.~Lewkowycz, and J.~Maldacena, {\it Quantum corrections to
  holographic entanglement entropy},  {\em Journal of High Energy Physics} {\bf
  2013} (2013), no.~74 [\href{http://arxiv.org/abs/1307.2892}{{\tt
  arXiv:1307.2892}}].

\bibitem{LM}
A.~Lewkowycz and J.~Maldacena, {\it Generalized gravitational entropy},  {\em
  Journal of High Energy Physics} {\bf 2013} (2013) 90,
  [\href{http://arxiv.org/abs/1304.4926}{{\tt arXiv:1304.4926}}].

\bibitem{DLR}
X.~Dong, A.~Lewkowycz, and M.~Rangamani, {\it Deriving covariant holographic
  entanglement},  {\em Journal of High Energy Physics} {\bf 2016} (2016) 28,
  [\href{http://arxiv.org/abs/1607.07506}{{\tt arXiv:1607.07506}}].

\bibitem{CKNR}
B.~Czech, J.~L. Karczmarek, F.~Nogueira, and M.~V. Raamsdonk, {\it The gravity
  dual of a density matrix},  {\em Classical and Quantum Gravity} {\bf 29}
  (2012), no.~15 [\href{http://arxiv.org/abs/1204.1330}{{\tt
  arXiv:1204.1330}}].

\bibitem{HHLR}
M.~Headrick, V.~E. Hubeny, A.~Lawrence, and M.~Rangamani, {\it Causality \&
  holographic entanglement entropy},  {\em Journal of High Energy Physics} {\bf
  2014} (2014), no.~12 [\href{http://arxiv.org/abs/1408.6300}{{\tt
  arXiv:1408.6300}}].

\bibitem{maximin}
A.~C. Wall, {\it Maximin surfaces, and the strong subadditivity of the
  covariant holographic entanglement entropy},  {\em Classical and Quantum
  Gravity} {\bf 31} (2014) 225007, [\href{http://arxiv.org/abs/1211.3494}{{\tt
  arXiv:1211.3494}}].

\bibitem{JLMS}
D.~L. Jafferis, A.~Lewkowycz, J.~Maldacena, and S.~J. Suh, {\it Relative
  entropy equals bulk relative entropy},  {\em Journal of High Energy Physics}
  {\bf 2016} (2016), no.~6 4, [\href{http://arxiv.org/abs/1512.06431}{{\tt
  arXiv:1512.06431}}].

\bibitem{DHW}
X.~Dong, D.~Harlow, and A.~C. Wall, {\it Reconstruction of bulk operators
  within the entanglement wedge in gauge-gravity duality},  {\em Physical
  Review Letters} {\bf 117} (2016) 021601,
  [\href{http://arxiv.org/abs/1601.05416}{{\tt arXiv:1601.05416}}].

\bibitem{noisyDHW}
J.~Cotler, P.~Hayden, G.~Penington, G.~Salton, B.~Swingle, and M.~Walter, {\it
  Entanglement wedge reconstruction via universal recovery channels},  {\em
  Physical Review X} {\bf 9} (2019) 031011,
  [\href{http://arxiv.org/abs/1704.05839}{{\tt arXiv:1704.05839}}].

\bibitem{may2019quantum}
A.~May, {\it Quantum tasks in holography},  {\em Journal of High Energy
  Physics} {\bf 2019} (2019), no.~233
  [\href{http://arxiv.org/abs/1902.06845}{{\tt arXiv:1902.06845}}].

\bibitem{EE-EE1}
M.~Nozaki, T.~Numasawa, A.~Prudenziati, and T.~Takayanagi, {\it Dynamics of
  entanglement entropy from {Einstein} equation},  {\em Physical Review D} {\bf
  88} (2013), no.~2 [\href{http://arxiv.org/abs/1304.7100}{{\tt
  arXiv:1304.7100}}].

\bibitem{EE-EE2}
N.~Lashkari, M.~B. McDermott, and M.~Van~Raamsdonk, {\it Gravitational dynamics
  from entanglement ``thermodynamics''},  {\em Journal of High Energy Physics}
  {\bf 2014} (2014), no.~4 [\href{http://arxiv.org/abs/1308.3716}{{\tt
  arXiv:1308.3716}}].

\bibitem{EE-EE3}
B.~Swingle and M.~V. Raamsdonk, {\it Universality of gravity from
  entanglement},  \href{http://arxiv.org/abs/1405.2933}{{\tt arXiv:1405.2933}}.

\bibitem{EE-EE4}
T.~Faulkner, F.~M. Haehl, E.~Hijano, O.~Parrikar, C.~Rabideau, and
  M.~Van~Raamsdonk, {\it Nonlinear gravity from entanglement in conformal field
  theories},  {\em Journal of High Energy Physics} {\bf 2017} (2017), no.~8
  [\href{http://arxiv.org/abs/1705.03026}{{\tt arXiv:1705.03026}}].

\bibitem{EE-EE5}
F.~M. Haehl, E.~Hijano, O.~Parrikar, and C.~Rabideau, {\it Higher curvature
  gravity from entanglement in conformal field theories},  {\em Physical Review
  Letters} {\bf 120} (2018), no.~20
  [\href{http://arxiv.org/abs/1712.06620}{{\tt arXiv:1712.06620}}].

\bibitem{EE-EE6}
A.~Lewkowycz and O.~Parrikar, {\it The holographic shape of entanglement and
  {E}instein's equations},  {\em Journal of High Energy Physics} {\bf 2018}
  (2018), no.~5 [\href{http://arxiv.org/abs/1802.10103}{{\tt
  arXiv:1802.10103}}].

\bibitem{waldbook}
R.~M. Wald, {\em General {R}elativity}.
\newblock University of Chicago Press, 1984.

\bibitem{gao2017traversable}
P.~Gao, D.~L. Jafferis, and A.~C. Wall, {\it Traversable wormholes via a double
  trace deformation},  {\em Journal of High Energy Physics} {\bf 2017} (2017),
  no.~151 [\href{http://arxiv.org/abs/1608.05687}{{\tt arXiv:1608.05687}}].

\bibitem{susskind2018teleportation}
L.~Susskind and Y.~Zhao, {\it Teleportation through the wormhole},  {\em
  Physical Review D} {\bf 98} (2018), no.~4 046016,
  [\href{http://arxiv.org/abs/1707.04354}{{\tt arXiv:1707.04354}}].

\bibitem{maldacena2017diving}
J.~Maldacena, D.~Stanford, and Z.~Yang, {\it Diving into traversable
  wormholes},  {\em Fortschritte der Physik} {\bf 65} (2017), no.~5 1700034,
  [\href{http://arxiv.org/abs/1704.05333}{{\tt arXiv:1704.05333}}].

\bibitem{bennett1984proceedings}
``Quantum cryptography: {P}ublic key distribution and coin tossing.''

\bibitem{tomamichel2013monogamy}
M.~Tomamichel, S.~Fehr, J.~Kaniewski, and S.~Wehner, {\it A
  monogamy-of-entanglement game with applications to device-independent quantum
  cryptography},  {\em New Journal of Physics} {\bf 15} (2013), no.~10 103002,
  [\href{http://arxiv.org/abs/1210.4359}{{\tt arXiv:1210.4359}}].

\bibitem{Nezami:2016zni}
S.~Nezami and M.~Walter, {\it {Multipartite Entanglement in Stabilizer Tensor
  Networks}},  \href{http://arxiv.org/abs/1608.02595}{{\tt arXiv:1608.02595}}.

\bibitem{Cui:2018dyq}
S.~X. Cui, P.~Hayden, T.~He, M.~Headrick, B.~Stoica, and M.~Walter, {\it {Bit
  Threads and Holographic Monogamy}},
  \href{http://arxiv.org/abs/1808.05234}{{\tt arXiv:1808.05234}}.

\bibitem{Hayden:2018khn}
P.~Hayden and G.~Penington, {\it {Learning the Alpha-bits of Black Holes}},
  {\em JHEP} {\bf 12} (2019) 007, [\href{http://arxiv.org/abs/1807.06041}{{\tt
  arXiv:1807.06041}}].

\bibitem{Akers:2019wxj}
C.~Akers, A.~Levine, and S.~Leichenauer, {\it {Large Breakdowns of Entanglement
  Wedge Reconstruction}},  {\em Phys. Rev.} {\bf D100} (2019), no.~12 126006,
  [\href{http://arxiv.org/abs/1908.03975}{{\tt arXiv:1908.03975}}].

\bibitem{Penington:2019kki}
G.~Penington, S.~H. Shenker, D.~Stanford, and Z.~Yang, {\it {Replica wormholes
  and the black hole interior}},  \href{http://arxiv.org/abs/1911.11977}{{\tt
  arXiv:1911.11977}}.

\bibitem{QES}
N.~Engelhardt and A.~C. Wall, {\it Quantum extremal surfaces: Holographic
  entanglement entropy beyond the classical regime},  {\em Journal of High
  Energy Physics} {\bf 2015} (2015), no.~73
  [\href{http://arxiv.org/abs/1408.3203}{{\tt arXiv:1408.3203}}].

\bibitem{QFC}
R.~Bousso, Z.~Fisher, S.~Leichenauer, and A.~C. Wall, {\it A quantum focussing
  conjecture},  {\em Physical Review D} {\bf 93} (2016) 064044,
  [\href{http://arxiv.org/abs/1506.02669}{{\tt arXiv:1506.02669}}].

\bibitem{quantum-maximin}
C.~Akers, N.~Engelhardt, G.~Penington, and M.~Usatyuk, {\it Quantum maximin
  surfaces},  \href{http://arxiv.org/abs/1912.02799}{{\tt arXiv:1912.02799}}.

\bibitem{maximin2}
D.~Marolf, A.~C. Wall, and Z.~Wang, {\it Restricted maximin surfaces and {HRT}
  in generic black hole spacetimes},  {\em Journal of High Energy Physics} {\bf
  2019} (2019) 127, [\href{http://arxiv.org/abs/1901.03879}{{\tt
  arXiv:1901.03879}}].

\bibitem{gerochDOD}
R.~Geroch, {\it Domain of dependence},  {\em Journal of Mathematical Phys.}
  {\bf 11} (1970) 437.

\bibitem{raychaudhuri}
A.~Raychaudhuri, {\it Relativistic cosmology. {I}},  {\em Physical Review} {\bf
  98} (1955) 1123.

\bibitem{boussobound}
R.~Bousso, {\it A covariant entropy conjecture},  {\em Journal of High Energy
  Physics} {\bf 1999} (1999), no.~07
  [\href{http://arxiv.org/abs/hep-th/9905177}{{\tt hep-th/9905177}}].

\bibitem{wittensingularities}
E.~Witten, {\it Light rays, singularities, and all that},
  \href{http://arxiv.org/abs/1901.03928}{{\tt arXiv:1901.03928}}.

\bibitem{sorce2019cutoffs}
J.~Sorce, {\it Holographic entanglement entropy is cutoff-covariant},
  \href{http://arxiv.org/abs/1908.02297}{{\tt arXiv:1908.02297}}.

\bibitem{chakraborty2015attack}
K.~Chakraborty and A.~Leverrier, {\it Attack strategies for position-based
  quantum cryptography based on the clifford hierarchy},  2015.

\bibitem{beigi2011simplified}
S.~Beigi and R.~K{\"o}nig, {\it Simplified instantaneous non-local quantum
  computation with applications to position-based cryptography},  {\em New
  Journal of Physics} {\bf 13} (2011), no.~9 093036,
  [\href{http://arxiv.org/abs/1101.1065}{{\tt arXiv:1101.1065}}].

\bibitem{dolev2019constraining}
K.~Dolev, {\it Constraining the doability of relativistic quantum tasks},
  \href{http://arxiv.org/abs/1909.05403}{{\tt arXiv:1909.05403}}.

\bibitem{ishizaka2008asymptotic}
S.~Ishizaka and T.~Hiroshima, {\it Asymptotic teleportation scheme as a
  universal programmable quantum processor},  {\em Physical Review Letters}
  {\bf 101} (2008), no.~24 240501, [\href{http://arxiv.org/abs/0807.4568}{{\tt
  arXiv:0807.4568}}].

\bibitem{horowitz1998ads}
G.~T. Horowitz and R.~C. Myers, {\it {AdS-CFT} correspondence and a new
  positive energy conjecture for general relativity},  {\em Physical Review D}
  {\bf 59} (1998), no.~2 026005,
  [\href{http://arxiv.org/abs/hep-th/9808079}{{\tt hep-th/9808079}}].

\bibitem{LCC1}
N.~Engelhardt and G.~T. Horowitz, {\it Towards a reconstruction of general bulk
  metrics},  {\em Classical and Quantum Gravity} {\bf 34} (2016), no.~1 015004,
  [\href{http://arxiv.org/abs/1605.01070}{{\tt arXiv:1605.01070}}].

\bibitem{LCC2}
N.~Engelhardt and G.~T. Horowitz, {\it Recovering the spacetime metric from a
  holographic dual},  \href{http://arxiv.org/abs/1612.00391}{{\tt
  arXiv:1612.00391}}.

\bibitem{rigidity1}
N.~Bao, C.~Cao, S.~Fischetti, and C.~Keeler, {\it Towards bulk metric
  reconstruction from extremal area variations},  {\em Classical and Quantum
  Gravity} {\bf 36} (2019), no.~18 185002,
  [\href{http://arxiv.org/abs/1904.04834}{{\tt arXiv:1904.04834}}].

\bibitem{fuchs1999cryptographic}
C.~A. Fuchs and J.~Van De~Graaf, {\it Cryptographic distinguishability measures
  for quantum-mechanical states},  {\em IEEE Transactions on Information
  Theory} {\bf 45} (1999), no.~4 1216--1227,
  [\href{http://arxiv.org/abs/quant-ph/9712042}{{\tt quant-ph/9712042}}].

\bibitem{wilde2013quantum}
M.~M. Wilde, {\em Quantum information theory, second edition}.
\newblock Cambridge University Press, 2017.

\bibitem{ambainis2000private}
A.~Ambainis, M.~Mosca, A.~Tapp, and R.~De~Wolf, {\it Private quantum channels},
   in {\em Proceedings 41st Annual Symposium on Foundations of Computer
  Science}, pp.~547--553, IEEE, 2000.

\end{thebibliography}\endgroup

\end{document}